\documentclass[12pt]{article}
\usepackage{eurosym}

\usepackage{amsfonts}
\usepackage{xcolor}
\usepackage{graphicx}
\usepackage{amsmath,amssymb,amsthm,bm}
\usepackage{times}
\usepackage[letterpaper,hmarginratio=1:1,vmarginratio=1:1,heightrounded,tmargin=1in,lmargin=1in,footnotesep=.6cm]{geometry}
\usepackage{mathptmx}
\usepackage{accents}

\usepackage[longnamesfirst,round]{natbib}
\def\citeapos#1{\citeauthor{#1}'s (\citeyear{#1})}

\usepackage{tikz}

\definecolor{Blue}{rgb}{0,0,0.8}
\usepackage[colorlinks,linkcolor=Blue,urlcolor=Blue,citecolor=Blue]{hyperref}

\graphicspath{{./graphs/}{C:/Users/jgm568/Dropbox/StrategicSampleSelection/EMA_SUBMISSION_3/graphs/}{C:/Users/okopso/Dropbox/StrategicSampleSelection/EMA_SUBMISSION_3/graphs/}{C:/Users/Ottaviani/Dropbox/StrategicSampleSelection/EMA_SUBMISSION_3/graphs/}{C:/Dropbox/StrategicSampleSelection/EMA_SUBMISSION_3/graphs/}{C:/Users/jgm568/Dropbox/StrategicSampleSelection/EMA_SUBMISSION_3/graphs/}}
\theoremstyle{plain}

\newtheorem{theoremone}{Theorem}

\newtheorem{theoremonestar}{Theorem}

\newtheorem{theoremtwo}{Theorem}

\newtheorem{theoremtwostar}{Theorem}

\newtheorem{theoremAFA}{Theorem}

\newtheorem{theoremfour}{Theorem}

\newtheorem{theoremfourstar}{Theorem}

\newtheorem{theoremfive}{Theorem}

\newtheorem{theoremfivestar}{Theorem}

\newtheorem{theoremsix}{Theorem}

\newtheorem{lemma}{Lemma}
\newtheorem{proposition}{Proposition}

\theoremstyle{definition}

\input{tcilatex}

\renewcommand{\geq}{\geqslant}
\renewcommand{\leq}{\leqslant}
\newcommand{\R}{\mathbb{R}}
\newcommand{\ve}{\varepsilon}
\renewenvironment{proof}[1][\proofname]{{\noindent\bfseries #1. }}{\hfill$\qedsymbol$}

\newcommand{\tconv}{\, \xrightarrow{\; P_\theta \;} \, }
\newcommand{\tconvp}{\, \xrightarrow{\; P_{\theta'} \;} \, }
\newcommand{\tconvd}{\, \xrightarrow{\; \mathrm{d}_{\theta} \;} \, }
\newcommand{\tdist}{\, \overset{\mathrm{d}_{\theta}}{=} \, }
\newcommand{\E}{\mathbb{E}}
\newcommand{\Var}{\mathrm{Var}}

\newcommand{\AND}{\qquad\text{and}\qquad}
\newcommand{\IMPLIES}{\qquad\Longrightarrow\qquad}

\begin{document}

\author{Alfredo Di Tillio\thanks{%
Department of Economics and IGIER, Bocconi University, Via Roberto Sarfatti
25, 20136 Milan, Italy. Phone: +39--02--5836--5422. E-mail: \href{mailto:alfredo.ditillio@unibocconi.it}%
{alfredo.ditillio@unibocconi.it}.} \and Marco Ottaviani\thanks{%
Department of Economics and IGIER, Bocconi University, Via Roberto Sarfatti
25, 20136 Milan, Italy. Phone: +39--02--5836--3385. E-mail: \href{mailto:marco.ottaviani@unibocconi.it}%
{marco.ottaviani@unibocconi.it}.} \and Peter Norman S\o rensen\thanks{%
Department of Economics, University of Copenhagen, \O ster Farimagsgade 5,
Building 26, DK--1353 Copenhagen K, Denmark. Phone: +45--3532--3056. E-mail: 
\href{mailto:peter.sorensen@econ.ku.dk}{peter.sorensen@econ.ku.dk}.}}
\title{\textbf{Information Comparison of Order Statistics,\\ \Large with Applications to Auctions and Voting\thanks{We thank, without implication, seminar participants at Universities of Zurich, Northwestern, Aarhus, Oslo and workshop audiences at 2026 NBER CEME Decentralization Conference, 2026 Barcelona Forum, CISEI 2026 in Anacapri, SAET-EWET 2026 in Venice for helpful comments. Financial support from the European Research Council (HEUROPE 2022 ADG, GA No. 101055294) is gratefully acknowledged.}}}
\date{\today}
\maketitle

\begin{abstract}
We compare the informativeness of order statistics in a sample of conditionally independent draws from a distribution $F(x|\theta)$ as the sample size $n$ increases. The $k$-th highest of $n+1$ draws is more accurate than the $k$-th highest of $n$ if and only if the cumulative reverse hazard $-\log F(x|\theta)$ is log-supermodular. Symmetrically, the $k$-th lowest is more accurate if and only if the cumulative hazard $-\log(1-F(x|\theta))$ is log-supermodular. Reversals are exceptional, occurring only for experiments that are, up to increasing transformations, exponential location experiments. In large samples, middle order statistics are asymptotically fully informative, while bounded lower and upper ranks require unbounded informativeness tail conditions. When full learning fails, bounded ranks converge to location experiments, and more central ranks are Blackwell more informative. Extending the analysis from scalar order statistics to blocks of selected data, we obtain multidimensional comparisons under log-supermodularity of hazard rates. The results unify and extend information-aggregation results in auctions and provide a new order-statistic approach to strategic voting.

\bigskip

\noindent \emph{Keywords:} Accuracy, comparison of experiments, order statistics, auctions, information aggregation, strategic voting, data selection.

\medskip

\noindent\emph{JEL codes:} D82, D83, C72, C90
\end{abstract}

\setcounter{page}{0} \thispagestyle{empty}

\newpage

\section{Introduction\label{SEC:INTRO}}

How does the information contained in the equilibrium price of an auction depend on the number of bidders and the number of objects for sale? How does jury size affect the probability of a correct decision when jurors vote strategically? More generally, how does data selection affect the
informativeness of a sample? 

Behind each of these questions lies the same statistical structure. An unknown state of the world---the true value of an asset or the culpability
of a defendant---generates $n$ signals: one for each of the $n$ bidders in an auction, one for each of the $n$ voters on a jury, or, more generally, one for each unit in a sample. These signals are independent of one another once we condition on the true state. 

Our analysis takes a monotonic approach, first asking how the information content of the equilibrium price in an auction or the probability of a
correct verdict changes when one further signal\ (another competing bidder or another juror) is added to the sample. Only then do we ask what happens as the number of signals grows without bound. 

This is precisely the two-step approach pioneered by \citet{Condorcet} in his \emph{Essai}, where he first asks how a jury's probability of reaching a correct verdict changes as jurors are added. In a second stage, Condorcet lets the number of jurors grow to infinity, concluding that a sufficiently large jury could attain ``as great a probability as one wishes'' (p.~9) of a correct verdict. Condorcet's analysis already carries this two-step logic through to an exact, closed-form answer, for jurors with binary signals who vote sincerely. The conditions under which Condorcet's conclusions hold for more general signals, when jurors vote strategically rather than sincerely, and for general voting rules and
arbitrary jury size, are far from obvious. 

The keystone of our analysis is a single and unifying observation: in both voting and auctions, the relevant object is an order statistic. When a
decision requires at least $k$ votes out of $n$ to pass, the outcome is determined by the $k$-th highest of the $n$ signals---the signal of the
pivotal voter. In an auction, likewise, when $n$ bidders compete for $k$ objects, the price is set by a pivotal bid that is itself governed by the $k$-th highest order statistic of the $n$ signals. Writing $X_{1},\dots,X_{n}$ for the $n$ conditionally independent signals and $X_{k,n}$ for the $k$-th highest among them, the key object throughout this paper is the information content of order statistic $X_{k,n}$ as $n$ varies.

To compare this information content rigorously, we leverage \citeapos{Lehmann} accuracy order. This order characterizes the value of information for preferences in the general interval dominance ordered (IDO) class introduced by \citet{Quah-Strulovici}, encompassing monotone decision problems \citep{Karlin-Rubin} and single-crossing preferences \citep{Milgrom-Shannon}. Thus, rather than comparing particular prices or voting outcomes, we compare the underlying experiments generated by the corresponding order statistics. Our only maintained assumption is the monotone likelihood ratio property: higher signal realizations are stronger evidence for higher states. Under this assumption, order-statistic comparisons translate directly into economic comparative statics in auctions, voting, and selected-data environments in general.

The first part of the paper studies finite-sample comparisons. Starting from $X_{k,n}$, adding one more draw produces two adjacent order statistics. If the additional draw lies below the old statistic, there is more maximal selection: the rank from the top of the new statistic $X_{k,n+1}$ is unchanged, but the statistic is selected from a larger sample. If the additional draw lies above the old statistic, there is more minimal selection: the new statistic $X_{k+1,n+1}$ moves one step away from the top. Theorem~\ref{THM:ACCORDION1} shows that the second comparison increases accuracy if and only if the cumulative hazard $H(x|\theta)=-\log(1-F(x|\theta))$ is log-supermodular. Theorem~\ref{THM:ACCORDION2} shows that the first comparison increases accuracy if and only if the cumulative reverse hazard $R(x|\theta)=-\log F(x|\theta)$ is log-supermodular. Thus the two adjacent ways of adding information are governed by the two tails of the signal distribution: lower-tail comparisons by $H$, upper-tail comparisons by $R$.

Theorems~\ref{THM:ACCORDION1} and \ref{THM:ACCORDION2} generalize the comparison in \citet{SSS}, where the question was whether observing the maximum $X_{1,n}$ from a larger sample is more informative. Here the same question is asked at every order statistic. Away from the boundary $k=1$, however, the negative counterpart is not obtained simply by reversing log-supermodularity. Theorems~\ref{THM:EXP1} and \ref{THM:EXP2} show that global decreasing patterns are exceptional. If more maximal selection decreases accuracy at all nodes, then the experiment must be, up to increasing transformations of the observation and the state, an exponential-noise location experiment. Symmetrically, if more minimal selection decreases accuracy at all nodes, then the experiment must be a reflected-exponential-noise location experiment. Thus, except for these special cases, adding a draw improves information in at least one of the two adjacent directions.

The second part of the paper studies large samples. Given a sequence $k_n$, Theorem~\ref{THM:AFA} characterizes when observing $X_{k_n,n}$ asymptotically reveals the state. If both $k_n$ and $n-k_n$ diverge, the statistic is asymptotically an intermediate quantile. Sampling noise around that quantile vanishes, and a minimal identifiability condition on the family $F(\cdot|\theta)$ is necessary and sufficient for asymptotic full accuracy. If $k_n$ is bounded, the statistic is one of the few highest observations, and full learning requires the upper-tail unbounded informativeness condition of \citet{Milgrom79}: sufficiently high signals must make likelihood ratios arbitrarily favorable to high states. Symmetrically, if $n-k_n$ is bounded, the statistic is one of the few lowest observations, and full learning requires the corresponding lower-tail condition.

When full learning fails at the extremes, Theorems~\ref{THM:CLASS1} and \ref{THM:CLASS2} characterize the residual limit information. For a fixed $r\geq1$, the lower-rank statistic $X_{n-r+1,n}$ asymptotically identifies the lower-tail equivalence class containing the true state. Within that class, after a natural logarithmic normalization, the limit experiment has a location form---and the fixed upper-rank statistic $X_{r,n}$, symmetrically, the reflected form. Theorems~\ref{THM:BLACKWELL1} and \ref{THM:BLACKWELL2} then compare these limit experiments as the fixed extreme rank $r$ varies. More central (higher $r$) extreme ranks are Blackwell more informative, and as $r\to\infty$ the residual noise vanishes.

Exploiting the multidimensional accuracy framework of \citet{SSS}, we present the counterpart of our one-dimensional results for vectors of selected data. Blocks of order statistics must be analyzed recursively: conditional on a coordinate in the block, the next coordinate is an order statistic (either the minimum or the maximum) of the remaining observations drawn from a truncated distribution. Therefore the relevant objects are not only the cumulative hazards $R$ and $H$, but their increments over arbitrary truncation intervals. The natural sufficient conditions strengthen from log-supermodularity of cumulative hazards to log-supermodularity of the reverse hazard rate $f/F$ and the hazard rate $f/(1-F)$. This block perspective is relevant for committee selection, for instance under peremptory challenge, where the remaining members of a committee correspond to a block of consecutive central order statistics.

Applying our results to auctions, we unify, qualify and extend all previous results in the information aggregation literature. In symmetric affiliated-value multi-unit auctions, equilibrium bids are increasing in signals. Hence the lowest winning bid and the highest rejected bid are monotone transformations of order statistics. Holding fixed the number of objects and adding one bidder compares $X_{k,n}$ with $X_{k,n+1}$, so price informativeness is governed by the cumulative reverse hazard $R$. Increasing both the number of bidders and the number of objects compares $X_{k,n}$ with $X_{k+1,n+1}$, so the relevant condition is log-supermodularity of the cumulative hazard $H$. These finite-auction comparative statics complement the information-aggregation results of \citet{Wilson77} and \citet{Milgrom79}. In the large-auction limit, our result recovers the double-largeness logic of \citet{Pesendorfer-Swinkels-1997}: if both the number of objects and the number of losing bidders diverge, the price is an intermediate order statistic and aggregates information under identification alone. If one side of the market remains bounded, the price is an extreme order statistic, and full aggregation requires the corresponding tail condition.

Finally, we apply our results to strategic voting. Given a rule requiring at least $k$ votes for conviction, a symmetric equilibrium convicts if and only if the $k$-th highest signal exceeds the equilibrium cutoff. Thus, as observed by \citet{DOS-voting}, the jury decision coincides with the decision that the pivotal voter, the voter with signal $X_{k,n}$, would take after conditioning on being pivotal. This key observation connects strategic voting to order-statistics experiments. Adding one juror while keeping fixed the number of votes required for conviction is a more-maximal-selection comparison; adding one juror while also raising the quota is a more-minimal-selection comparison. In large juries, proportional rules aggregate information under identification. Extreme rules instead rely on either tail of the jurors' signal distribution: unanimity for conviction, for example, makes acquittal depend on a bounded number of low signals, so asymptotic correctness requires lower-tail unbounded informativeness. This recasts the lessons of \citet{Feddersen-Pesendorfer} and \citet{Duggan-Martinelli} in order-statistic terms.

A natural further question, especially relevant in voting applications, is how to compare order statistics for a given sample size---changing $k$ while keeping $n$ fixed. For example, in a jury of size $n$, changing the quota from $k$ to $k+1$ compares the adjacent order statistics $X_{k,n}$ and $X_{k+1,n}$. Such comparisons are economically important, because they correspond to making conviction or acceptance more demanding without changing the size of the committee. They are also substantially different from the global comparisons studied here. Adjacent order statistics with the same sample size are rarely comparable in terms of \citeapos{Lehmann} order. In general, neither one is more informative than the other for all monotone decision problems. The comparison is therefore typically local, depending on the specific trade-off between the two types of error. We develop this local notion of accuracy, together with the associated equilibrium and stability analysis, in the companion paper \citet{DOS-voting}. The present paper focuses instead on comparisons generated by increasing sample size and on the large-sample limits, where global comparisons are more broadly applicable.

The paper is organized as follows. Section~\ref{SEC:SETUP} presents the setup and Lehmann's accuracy order. Section~\ref{SEC:ACCORDION} develops the finite-sample order-statistic comparisons. Section~\ref{SEC:EXTREME} studies large samples, characterizes asymptotic full accuracy, and derives the nondegenerate limit experiments for fixed extreme ranks. Section~\ref{SEC:MULTI} extends the analysis to multidimensional experiments.  Sections~\ref{SEC:AUCTIONS} and \ref{SEC:VOTING} apply the results to auctions and voting. Proofs are collected in Appendix~\ref{APP:PROOFS}.

\section{Setup\label{SEC:SETUP}}

Consider a family of experiments indexed by a state $\theta\in\Theta\subseteq \R$, where $\Theta$ is either a finite set or a (possibly unbounded) interval. Conditional on $\theta$, the random variables $X_1,X_2,\ldots$ are i.i.d.~with cumulative distribution function $F(\cdot|\theta)$. We assume that $F(x|\theta)$ is continuous in $\theta$ and that, for every $\theta$, $F(\cdot|\theta)$ is absolutely continuous with density $f(\cdot|\theta)$, positive on the interior of its support. We write
\[
\underline x_\theta:=\inf\{x:F(x|\theta)>0\},
\qquad
\overline x_\theta:=\sup\{x:F(x|\theta)<1\}
\]
for the lower and upper bounds of the support under $\theta$. Throughout, $X_{k,n}$ denotes the $k$-th highest draw in the size-$n$ sample $X_1,\ldots,X_n$, so that $X_{1,n}\geq X_{2,n}\geq \cdots \geq X_{n,n}$.

Besides the regularity assumptions listed above, our only maintained economic assumption is the monotone likelihood ratio property, that is, for every pair of states $\theta$ and $\theta'$,
\begin{equation}
\label{EQ:MLR}
\tag{MLR}
\theta'>\theta
\IMPLIES
\frac{f(x|\theta')}{f(x|\theta)}
\quad\text{is increasing in $x$.}
\end{equation}
Since log-supermodularity is preserved by integration, \eqref{EQ:MLR} implies
\begin{equation}
\label{EQ:MLR2}
\tag{MLR2}
\theta'>\theta
\IMPLIES
\frac{F(x|\theta')}{F(x|\theta)}
\ \text{and}\
\frac{1-F(x|\theta')}{1-F(x|\theta)}
\ \text{are both increasing in }x.
\end{equation}
Moreover, since the likelihood-ratio order implies the usual stochastic order, i.e.~first-order stochastic dominance, \eqref{EQ:MLR} also implies
\begin{equation}
\label{EQ:D}
\tag{D}
\theta'>\theta
\quad\Longrightarrow\quad
F(x|\theta')\leq F(x|\theta)
\quad\text{for all }x.
\end{equation}

Our comparisons use \citeapos{Lehmann} accuracy order. Thus, for any rank-size pairs $k\leq n$ and $k'\leq n'$, we say $X_{k',n'}$ is \emph{more accurate} than $X_{k,n}$ if for all states $\theta'>\theta$ and realizations $x,x'\in\R$,
\[
\Pr\nolimits_\theta(X_{k',n'}\leq x')
=
\Pr\nolimits_\theta(X_{k,n}\leq x)
\IMPLIES
\Pr\nolimits_{\theta'}(X_{k',n'}\leq x')
\leq
\Pr\nolimits_{\theta'}(X_{k,n}\leq x)
\]
Intuitively, consider testing state $\theta$ against the higher state $\theta'$, with the higher action chosen when the observation exceeds a cutoff. If the cutoff $x'$ for $X_{k',n'}$ is chosen to match the type-I error induced by cutoff $x$ for $X_{k,n}$, then the inequality says that $X_{k',n'}$ induces a smaller type-II error. Thus, at every false-positive rate, the more accurate statistic gives a smaller false-negative rate. The accuracy order is invariant to strictly increasing transformations of the observation and of the state, and under \eqref{EQ:MLR} it is the appropriate comparison for monotone decision problems in the \emph{interval dominance order} family---see \citet{Quah-Strulovici} and \citet{SSS}.

\section{Minimal and Maximal Selection\label{SEC:ACCORDION}}

We begin with finite-sample comparisons. Fixing an order statistic $X_{k,n}$, we ask what happens to its information content when one additional draw is added, that is, $n$ increases by one. If the new draw lies below $X_{k,n}$, the statistic becomes $X_{k,n+1}$. If it lies above, the statistic becomes $X_{k+1,n+1}$. This section characterizes when each of these two movements increases or decreases accuracy.

\subsection{Increasing Accuracy}

It is useful to organize our comparisons in a triangular diagram where each node represents an order statistic $X_{k,n}$. Moving from row $n$ to row $n+1$, if the additional draw falls below $X_{k,n}$ the observation becomes $X_{k,n+1}$ and we say that there is \emph{more maximal selection}: the observation is still the $k$-th highest, but the number of (unseen) realizations below it has grown from $n-k$ to $n-k+1$. If the additional draw falls above $X_{k,n}$, the observation becomes $X_{k+1,n+1}$, and we say that there is \emph{more minimal selection}: the observation becomes the $(k+1)$-th highest, as the number of (unseen) realizations below it remains the same ($n-k$) but those above it have increased from $k-1$ to $k$. We refer to this triangular system of adjacent comparisons as the \emph{accordion}.

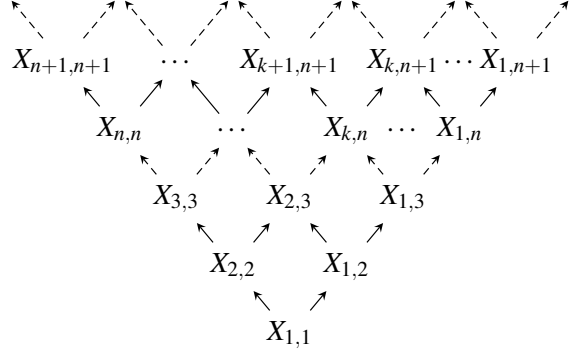
\begin{figure}[t]
\centering
\begingroup

\def\dx{1.5} 
\def\dy{0.9} 

\begin{tikzpicture}[
    x=1cm,
    y=1cm,
    >=stealth,
    every node/.style={font=\small, inner sep=1.5pt},
    arr/.style={->, line width=0.45pt, shorten <=3pt, shorten >=3pt},
    contarr/.style={->, densely dashed, line width=0.45pt, shorten <=3pt, shorten >=3pt}
]


\node (X11) at (0,0) {$X_{1,1}$};

\node (X22) at ({-0.5*\dx},{1*\dy}) {$X_{2,2}$};
\node (X12) at ({ 0.5*\dx},{1*\dy}) {$X_{1,2}$};

\node (X33) at ({-1*\dx},{2*\dy}) {$X_{3,3}$};
\node (X23) at ({ 0*\dx},{2*\dy}) {$X_{2,3}$};
\node (X13) at ({ 1*\dx},{2*\dy}) {$X_{1,3}$};


\node (Xnn)    at ({-1.5*\dx},{3*\dy}) {$X_{n,n}$};
\node (dotsnL) at ({-0.5*\dx},{3*\dy}) {$\cdots$};
\node (Xkn)    at ({ 0.5*\dx},{3*\dy}) {$X_{k,n}$};
\node (dotsnR) at ({ 1.0*\dx},{3*\dy}) {$\cdots$};
\node (X1n)    at ({ 1.5*\dx},{3*\dy}) {$X_{1,n}$};


\node (Xnp1np1) at ({-2.0*\dx},{4*\dy}) {$X_{n+1,n+1}$};
\node (dotsp1L) at ({-1.0*\dx},{4*\dy}) {$\cdots$};
\node (Xkp1np1) at ({ 0.0*\dx},{4*\dy}) {$X_{k+1,n+1}$};
\node (Xknp1)   at ({ 1.0*\dx},{4*\dy}) {$X_{k,n+1}$};
\node (dotsp1R) at ({ 1.5*\dx},{4*\dy}) {$\cdots$};
\node (X1np1)   at ({ 2.0*\dx},{4*\dy}) {$X_{1,n+1}$};

\coordinate (U0) at ({-2.5*\dx},{5*\dy});
\coordinate (U1) at ({-1.5*\dx},{5*\dy});
\coordinate (U2) at ({-0.5*\dx},{5*\dy});
\coordinate (U3) at ({ 0.5*\dx},{5*\dy});
\coordinate (U4) at ({ 1.5*\dx},{5*\dy});
\coordinate (U5) at ({ 2.5*\dx},{5*\dy});


\draw[arr] (X11) -- (X22);
\draw[arr] (X11) -- (X12);

\draw[arr] (X22) -- (X33);
\draw[arr] (X22) -- (X23);

\draw[arr] (X12) -- (X23);
\draw[arr] (X12) -- (X13);


\draw[contarr] (X33) -- (Xnn);
\draw[contarr] (X33) -- (dotsnL);

\draw[contarr] (X23) -- (dotsnL);
\draw[contarr] (X23) -- (Xkn);

\draw[contarr] (X13) -- (Xkn);
\draw[contarr] (X13) -- (X1n);


\draw[arr] (Xnn) -- (Xnp1np1);
\draw[arr] (Xnn) -- (dotsp1L);

\draw[arr] (dotsnL) -- (dotsp1L);
\draw[arr] (dotsnL) -- (Xkp1np1);

\draw[arr] (Xkn) -- (Xkp1np1);
\draw[arr] (Xkn) -- (Xknp1);

\draw[arr] (X1n) -- (Xknp1);

\draw[arr] (X1n) -- (X1np1);


\draw[contarr] (Xnp1np1) -- (U0);
\draw[contarr] (Xnp1np1) -- (U1);

\draw[contarr] (dotsp1L) -- (U1);
\draw[contarr] (dotsp1L) -- (U2);

\draw[contarr] (Xkp1np1) -- (U2);
\draw[contarr] (Xkp1np1) -- (U3);

\draw[contarr] (Xknp1) -- (U3);
\draw[contarr] (Xknp1) -- (U4);

\draw[contarr] (X1np1) -- (U4);
\draw[contarr] (X1np1) -- (U5);

\end{tikzpicture}

\endgroup
\caption{The accordion when both more minimal and more maximal selection increase accuracy.}
\label{FIG:ACCORDION1}
\end{figure}

Figure~\ref{FIG:ACCORDION1} depicts the case in which accuracy increases as the accordion expands in either direction---an arrow means that accuracy increases in the direction of the arrow. Consider first the branches corresponding to more minimal selection, from $X_{k,n}$ to $X_{k+1,n+1}$. The answer in this case is governed by the cumulative hazard of the basic distribution,
\[
H(x|\theta):=-\log\big[1-F(x|\theta)\big].
\]

\medskip

\begin{theoremone}
\label{THM:ACCORDION1}
More minimal selection increases accuracy if and only if $H$ is log-supermodular.
\end{theoremone}

The counterpart for maximal selection, Theorem~\ref{THM:ACCORDION2} below, follows from a simple reflection argument. Replacing each variable $X_i$ with $-X_i$ and each state $\theta$ with $-\theta$ turns maximal selection in the original experiment into minimal selection in the reflected experiment, and the cumulative hazard of the reflected experiment is the cumulative reverse hazard of the original experiment,
\[
R(x|\theta):=-\log F(x|\theta).
\]

\begin{theoremonestar}
\label{THM:ACCORDION2}
More maximal selection increases accuracy if and only if $R$ is log-supermodular.
\end{theoremonestar}

Theorems~\ref{THM:ACCORDION1} and \ref{THM:ACCORDION2} generalize the comparison in Theorem~2 of \citet{SSS}. In that paper, more maximal selection means observing the highest order statistic in a larger sample---moving from $X_{1,n}$ to $X_{1,n+1}$. Symmetrically, more minimal selection means observing the lowest order statistic in a larger sample---moving from $X_{n,n}$ to $X_{n+1,n+1}$. In this paper, the same questions are asked at every node of the accordion.

Differently from the special case ($k=1$ or $k=n$) considered in \citet{SSS}, however, log-submodularity of $H$ or $R$ does \emph{not} reverse the statement in Theorem~\ref{THM:ACCORDION1} or \ref{THM:ACCORDION2}. If $R$ is log-submodular, the maximum $X_{1,n}$ becomes less informative as $n$ grows; and if $H$ is log-submodular, the minimum $X_{n,n}$ becomes less informative as $n$ grows. But these boundary comparisons do not extend to intermediate order statistics. For example, with Gumbel noise the maximum $X_{1,n}$ is equally accurate for all $n$, while the other order statistics become more accurate as $n$ grows.

The next subsection identifies the exceptional cases in which a whole family of arrows does reverse direction. Those cases arise only when, up to monotone transformations of the observation and the state, the experiment has exponential (or reflected-exponential) noise.

\subsection{Decreasing Accuracy}

Theorems~\ref{THM:ACCORDION1} and \ref{THM:ACCORDION2} identify the conditions under which both families of arrows in the accordion point upward. A global reversal, with both more maximal and more minimal selection decreasing accuracy, is impossible. To see why, fix a node $X_{k,n}$ and consider a sample of size $n+1$, from which one draw is deleted at random. If the deleted draw is below $X_{k,n+1}$, the resulting $k$-th highest order statistic is $X_{k,n+1}$; if instead the deleted draw is among the $k$ highest, the resulting statistic is $X_{k+1,n+1}$. Hence $X_{k,n}$ is a state-independent mixture of its two children $X_{k,n+1}$ and $X_{k+1,n+1}$. Clearly, observing one of the children at random, while knowing which child it is, Blackwell dominates observing only the mixture. Therefore the parent cannot be more valuable than both children in every decision problem, except in payoff-equivalent degenerate cases. Thus a nontrivial accordion cannot have all arrows reversed.

\medskip
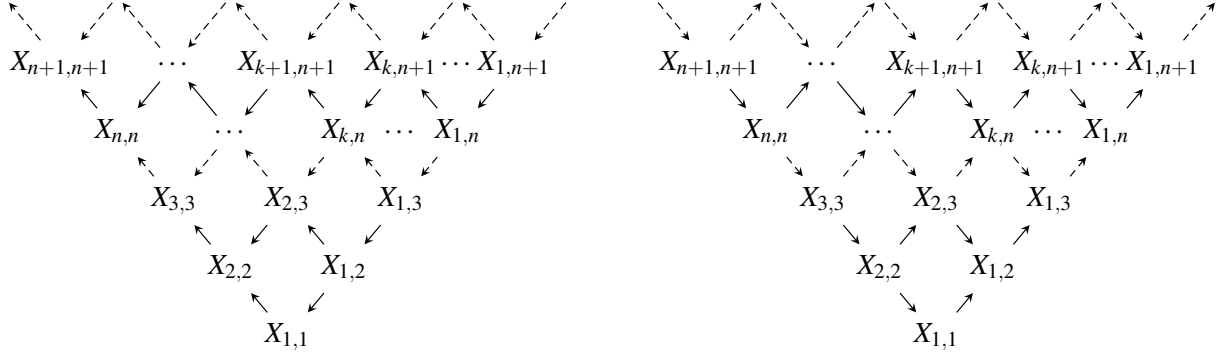
\begin{figure}[t]
\centering
\begin{minipage}[t]{0.48\textwidth}
\centering
\begingroup

\def\dx{1.5} 
\def\dy{0.9}   

\begin{tikzpicture}[
    x=1cm,
    y=1cm,
    >=stealth,
    every node/.style={font=\small, inner sep=1.5pt},
    arr/.style={->, line width=0.45pt, shorten <=3pt, shorten >=3pt},
    contarr/.style={->, densely dashed, line width=0.45pt, shorten <=3pt, shorten >=3pt}
]


\node (X11) at (0,0) {$X_{1,1}$};

\node (X22) at ({-0.5*\dx},{1*\dy}) {$X_{2,2}$};
\node (X12) at ({ 0.5*\dx},{1*\dy}) {$X_{1,2}$};

\node (X33) at ({-1*\dx},{2*\dy}) {$X_{3,3}$};
\node (X23) at ({ 0*\dx},{2*\dy}) {$X_{2,3}$};
\node (X13) at ({ 1*\dx},{2*\dy}) {$X_{1,3}$};


\node (Xnn)    at ({-1.5*\dx},{3*\dy}) {$X_{n,n}$};
\node (dotsnL) at ({-0.5*\dx},{3*\dy}) {$\cdots$};
\node (Xkn)    at ({ 0.5*\dx},{3*\dy}) {$X_{k,n}$};
\node (dotsnR) at ({ 1.0*\dx},{3*\dy}) {$\cdots$};
\node (X1n)    at ({ 1.5*\dx},{3*\dy}) {$X_{1,n}$};


\node (Xnp1np1) at ({-2.0*\dx},{4*\dy}) {$X_{n+1,n+1}$};
\node (dotsp1L) at ({-1.0*\dx},{4*\dy}) {$\cdots$};
\node (Xkp1np1) at ({ 0.0*\dx},{4*\dy}) {$X_{k+1,n+1}$};
\node (Xknp1)   at ({ 1.0*\dx},{4*\dy}) {$X_{k,n+1}$};
\node (dotsp1R) at ({ 1.5*\dx},{4*\dy}) {$\cdots$};
\node (X1np1)   at ({ 2.0*\dx},{4*\dy}) {$X_{1,n+1}$};

\coordinate (U0) at ({-2.5*\dx},{5*\dy});
\coordinate (U1) at ({-1.5*\dx},{5*\dy});
\coordinate (U2) at ({-0.5*\dx},{5*\dy});
\coordinate (U3) at ({ 0.5*\dx},{5*\dy});
\coordinate (U4) at ({ 1.5*\dx},{5*\dy});
\coordinate (U5) at ({ 2.5*\dx},{5*\dy});


\draw[arr] (X11) -- (X22);
\draw[arr] (X12) -- (X11);

\draw[arr] (X22) -- (X33);
\draw[arr] (X23) -- (X22);

\draw[arr] (X12) -- (X23);
\draw[arr] (X13) -- (X12);


\draw[contarr] (X33) -- (Xnn);
\draw[contarr] (dotsnL) -- (X33);

\draw[contarr] (X23) -- (dotsnL);
\draw[contarr] (Xkn) -- (X23);

\draw[contarr] (X13) -- (Xkn);
\draw[contarr] (X1n) -- (X13);


\draw[arr] (Xnn) -- (Xnp1np1);
\draw[arr] (dotsp1L) -- (Xnn);

\draw[arr] (dotsnL) -- (dotsp1L);
\draw[arr] (Xkp1np1) -- (dotsnL);

\draw[arr] (Xkn) -- (Xkp1np1);
\draw[arr] (Xknp1) -- (Xkn);

\draw[arr] (X1n) -- (Xknp1);

\draw[arr] (X1np1) -- (X1n);


\draw[contarr] (Xnp1np1) -- (U0);
\draw[contarr] (U1) -- (Xnp1np1);

\draw[contarr] (dotsp1L) -- (U1);
\draw[contarr] (U2) -- (dotsp1L);

\draw[contarr] (Xkp1np1) -- (U2);
\draw[contarr] (U3) -- (Xkp1np1);

\draw[contarr] (Xknp1) -- (U3);
\draw[contarr] (U4) -- (Xknp1);

\draw[contarr] (X1np1) -- (U4);
\draw[contarr] (U5) -- (X1np1);

\end{tikzpicture}

\endgroup
\end{minipage}%
\hfill
\begin{minipage}[t]{0.48\textwidth}
\centering
\begingroup

\def\dx{1.5} 
\def\dy{0.9}   

\begin{tikzpicture}[
    x=1cm,
    y=1cm,
    >=stealth,
    every node/.style={font=\small, inner sep=1.5pt},
    arr/.style={->, line width=0.45pt, shorten <=3pt, shorten >=3pt},
    contarr/.style={->, densely dashed, line width=0.45pt, shorten <=3pt, shorten >=3pt}
]


\node (X11) at (0,0) {$X_{1,1}$};

\node (X22) at ({-0.5*\dx},{1*\dy}) {$X_{2,2}$};
\node (X12) at ({ 0.5*\dx},{1*\dy}) {$X_{1,2}$};

\node (X33) at ({-1*\dx},{2*\dy}) {$X_{3,3}$};
\node (X23) at ({ 0*\dx},{2*\dy}) {$X_{2,3}$};
\node (X13) at ({ 1*\dx},{2*\dy}) {$X_{1,3}$};


\node (Xnn)    at ({-1.5*\dx},{3*\dy}) {$X_{n,n}$};
\node (dotsnL) at ({-0.5*\dx},{3*\dy}) {$\cdots$};
\node (Xkn)    at ({ 0.5*\dx},{3*\dy}) {$X_{k,n}$};
\node (dotsnR) at ({ 1.0*\dx},{3*\dy}) {$\cdots$};
\node (X1n)    at ({ 1.5*\dx},{3*\dy}) {$X_{1,n}$};


\node (Xnp1np1) at ({-2.0*\dx},{4*\dy}) {$X_{n+1,n+1}$};
\node (dotsp1L) at ({-1.0*\dx},{4*\dy}) {$\cdots$};
\node (Xkp1np1) at ({ 0.0*\dx},{4*\dy}) {$X_{k+1,n+1}$};
\node (Xknp1)   at ({ 1.0*\dx},{4*\dy}) {$X_{k,n+1}$};
\node (dotsp1R) at ({ 1.5*\dx},{4*\dy}) {$\cdots$};
\node (X1np1)   at ({ 2.0*\dx},{4*\dy}) {$X_{1,n+1}$};

\coordinate (U0) at ({-2.5*\dx},{5*\dy});
\coordinate (U1) at ({-1.5*\dx},{5*\dy});
\coordinate (U2) at ({-0.5*\dx},{5*\dy});
\coordinate (U3) at ({ 0.5*\dx},{5*\dy});
\coordinate (U4) at ({ 1.5*\dx},{5*\dy});
\coordinate (U5) at ({ 2.5*\dx},{5*\dy});


\draw[arr] (X22) -- (X11);
\draw[arr] (X11) -- (X12);

\draw[arr] (X33) -- (X22);
\draw[arr] (X22) -- (X23);

\draw[arr] (X23) -- (X12);
\draw[arr] (X12) -- (X13);


\draw[contarr] (Xnn) -- (X33);
\draw[contarr] (X33) -- (dotsnL);

\draw[contarr] (dotsnL) -- (X23);
\draw[contarr] (X23) -- (Xkn);

\draw[contarr] (Xkn) -- (X13);
\draw[contarr] (X13) -- (X1n);


\draw[arr] (Xnp1np1) -- (Xnn);
\draw[arr] (Xnn) -- (dotsp1L);

\draw[arr] (dotsp1L) -- (dotsnL);
\draw[arr] (dotsnL) -- (Xkp1np1);

\draw[arr] (Xkp1np1) -- (Xkn);
\draw[arr] (Xkn) -- (Xknp1);

\draw[arr] (Xknp1) -- (X1n);

\draw[arr] (X1n) -- (X1np1);


\draw[contarr] (U0) -- (Xnp1np1);
\draw[contarr] (Xnp1np1) -- (U1);

\draw[contarr] (U1) -- (dotsp1L);
\draw[contarr] (dotsp1L) -- (U2);

\draw[contarr] (U2) -- (Xkp1np1);
\draw[contarr] (Xkp1np1) -- (U3);

\draw[contarr] (U3) -- (Xknp1);
\draw[contarr] (Xknp1) -- (U4);

\draw[contarr] (U4) -- (X1np1);
\draw[contarr] (X1np1) -- (U5);

\end{tikzpicture}

\endgroup
\end{minipage}
\caption{Left: more minimal selection increases and more maximal selection decreases accuracy. Right: more minimal selection decreases and more maximal selection increases accuracy.}
\label{FIG:ACCORDION2}
\end{figure}

The only remaining global orientations are therefore mixed: one branch points upward and the other downward. Figure~\ref{FIG:ACCORDION2} displays the two possibilities. In the left panel, more minimal selection increases accuracy, while more maximal selection decreases it. In the right panel, the roles are reversed. The next two theorems show that these mixed patterns are exceptional. Assuming that more maximal selection decreases accuracy---that is, assuming that all downward-pointing arrows in the left panel of Figure~\ref{FIG:ACCORDION2} hold---is enough to obtain the entire left-panel pattern. Moreover, it implies that the experiment is, up to increasing transformations of the observation and the state, an exponential-noise location experiment. Conversely, every such experiment gives rise to the left-panel pattern.

\begin{theoremtwo}
\label{THM:EXP1}
The following are equivalent:
\begin{enumerate}
\item[(i)] More maximal selection decreases accuracy.
\item[(ii)] More minimal selection increases and more maximal selection decreases accuracy.
\item[(iii)] There exist strictly increasing functions $\theta\mapsto\lambda(\theta)$ and $x\mapsto \xi(x)$ such that
\[
 \xi(X_i) \tdist \lambda(\theta)+E_{i}
\]
where $E_{i}$ is an exponential random variable.
\end{enumerate}
\end{theoremtwo}

The proof of Theorem~\ref{THM:EXP1} shows that, in fact, each of (i)--(iii) is equivalent to the following weakening of (i): for every $n$, the minimum of $n$ draws, $X_{n,n}$ is more accurate than the second-lowest of $n+1$ draws, $X_{n,n+1}$. Analogously, in the mirror statement below, which gives the right-panel of Figure~\ref{FIG:ACCORDION2}, one can weaken condition (i) to the following: for every $n$, the maximum of $n$ draws, $X_{1,n}$, is more accurate than the second-highest of $n+1$ draws, $X_{2,n+1}$.

\begin{theoremtwostar}
\label{THM:EXP2}
The following are equivalent:
\begin{enumerate}
\item[(i)] More minimal selection decreases accuracy.
\item[(ii)] More minimal selection decreases and more maximal selection increases accuracy.
\item[(iii)] There exist a strictly increasing $\theta\mapsto\lambda(\theta)$ and a strictly increasing $x\mapsto\xi(x)$ such that
\[
\xi(X_i) \tdist \lambda(\theta)-E_{i}
\]
where $E_{i}$ is an exponential random variable.
\end{enumerate}
\end{theoremtwostar}

\subsection{Location Experiments}

In a location experiment, the random variables $X_i$ take the form $X_i=\theta+\ve_i$. The noise term $\ve_i$ is drawn from a distribution $F$ with density $f$, such that $F(x-\theta)=F(x|\theta)$. The conditions in Theorems~\ref{THM:ACCORDION1} and \ref{THM:ACCORDION2} reduce to shape restrictions on the two tails of $F$. Indeed, log-supermodularity of $H$ is equivalent to logconcavity of $-\log(1-F)$, while log-supermodularity of $R$ is equivalent to logconcavity of $-\log F$. Thus, both more minimal and more maximal selection increase accuracy e.g.~with normal or logistic noise.\footnote{Note that our comparisons do not have a stronger, Blackwell-dominance counterpart. \citet{Lehmann} himself remarks the shortcoming of Blackwell's order, that no non-normal location experiment can be better than a normal location experiment. This applies here when $X_{1,1}$ is normal, given that $X_{1,n}$ is not normal for $n\geq2$.} The Gumbel distribution $F(\ve)=\exp(-e^{-\ve})$ is a boundary case: the reverse hazard function $-\log F(\ve)=e^{-\ve}$ is loglinear. As a result, the maximum $X_{1,n}$ is equally accurate for all $n$, while the remaining same-rank comparisons are still accuracy-improving.

When specialized to location experiments, the proof of Theorem~\ref{THM:EXP1} forces the transformations $x\mapsto \xi(x)$ and $\theta\mapsto\lambda(\theta)$ in representation (iii) to be affine. Hence the noise itself must be a shifted and rescaled exponential: $X_i=\theta+b+E_i/\rho$ for some $b\in\R$ and $\rho>0$. This is exactly the case in which more maximal selection decreases accuracy while more minimal selection increases it. By reflection, Theorem~\ref{THM:EXP2} gives the opposite pattern: the only location experiments in which more minimal selection decreases and more maximal selection increases accuracy are those with shifted and rescaled reflected-exponential noise, $X_i=\theta-b-E_i/\rho$.

\section{Large Samples\label{SEC:EXTREME}}

In this section we characterize the information content of order statistics when the sample size $n$ grows large. Given a sequence $k_1,k_2,\ldots$ such that $k_n\leq n$ for every $n$, we give conditions on the basic distribution $F$ and on the sequence $k_n$ that are necessary and sufficient for asymptotic learning of the state, based on the observation of the statistic $X_{k_n,n}$ alone. In the complementary case where learning is less than full in the limit, we provide information comparisons among the nondegenerate limit experiments to which the order statistics converge.

\subsection{Asymptotic Full Accuracy}
Our formalization of asymptotic learning is standard. We say that a sequence of experiments $Y_n$ has \emph{asymptotic full accuracy} (\emph{AFA}) if the state can be consistently estimated from the experiment's observation, that is, if there exists a sequence of functions $\phi_n:\R\to\R$ such that, for every $\theta\in\Theta$,
\[
 \phi_n(Y_n) \, \tconv\, \theta.
\]

In our context, the sequence of experiments is $X_{k_n,n}$. Thus, AFA necessarily fails if the distribution of the variables $X_i$ at two different states coincide. In this section we therefore make the following minimal identification assumption: for every pair of states $\theta$ and $\theta'$,
\begin{equation}
\label{EQ:ID}\tag{ID}
\theta'\neq\theta
\qquad
\Longrightarrow
\qquad
F(\cdot|\theta')\neq F(\cdot|\theta).
\end{equation}
Observe that, pairing \eqref{EQ:MLR} with \eqref{EQ:ID}, the stochastic dominance \eqref{EQ:D} becomes strict:\footnote{To see this, suppose that \eqref{EQ:SD} does not hold. Then there exists $\underline{x}_\theta<x_0<\overline{x}_{\theta'}$ such that $F(x_0|\theta')=F(x_0|\theta)$. By \eqref{EQ:MLR2} the ratio $F(x|\theta')/F(x|\theta)$ is increasing, while by \eqref{EQ:D} it is at most one. Thus, $F(x|\theta')=F(x|\theta)$ for all $x\geq x_0$. Similarly, by \eqref{EQ:MLR2} the ratio $(1-F(x|\theta'))/(1-F(x|\theta))$ is increasing, while by \eqref{EQ:D} it is at least one. Thus, $1-F(x|\theta')=1-F(x|\theta)$ and hence $F(x|\theta')=F(x|\theta)$ for all $x\leq x_0$. This contradicts \eqref{EQ:ID}.} 
\begin{equation}
\label{EQ:SD}\tag{SD}
\theta'>\theta
\qquad
\Longrightarrow
\qquad
F(x|\theta') < F(x|\theta)
\quad\text{for all} \;\; \underline{x}_\theta<x<\overline{x}_{\theta'}.
\end{equation}

The following theorem shows that for middle order statistics, those such that both $k_n\to\infty$ and $n-k_n\to\infty$, no condition on $F$ beyond \eqref{EQ:MLR} and \eqref{EQ:ID} is needed for AFA. In particular, the median, and more generally all those such that $k_n/n$ converges to some number in (0,1), are fully informative in the limit. For bounded lower-rank or upper-rank statistics, instead, AFA depends on how the tails of $F$ at different states compare. For bounded lower-rank statistics, those such that the sequence $n-k_n$ is bounded, AFA obtains if and only if the following \emph{lower unbounded informativeness} condition holds: for every pair of states $\theta$ and $\theta'$,
\begin{align}
\label{EQ:LI}\tag{LI}
\theta'>\theta
\qquad
&\Longrightarrow
\qquad
\inf_ x\frac{f(x|\theta')}{f(x|\theta)}=0.
\end{align}
For bounded upper-rank statistics, those such that $k_n$ is bounded, AFA obtains if and only if the symmetric \emph{upper unbounded informativeness} condition holds:
\begin{align}
\label{EQ:UI}\tag{UI}
\theta'>\theta
\qquad
&\Longrightarrow
\qquad
\sup_ x\frac{f(x|\theta')}{f(x|\theta)}=\infty.
\end{align}

\begin{theoremAFA}
\label{THM:AFA} Assume \eqref{EQ:ID}.
\begin{enumerate}
\item[(i)] If both $k_n\to\infty$ and $n-k_n\to\infty$, then $X_{k_n,n}$ has AFA.
\item[(ii)] If $n-k_n$ is a bounded sequence, then $X_{k_n,n}$ has AFA if and only if \eqref{EQ:LI} holds.
\item[(iii)] If $k_n$ is a bounded sequence, then $X_{k_n,n}$ has AFA if and only if \eqref{EQ:UI} holds.
\end{enumerate}
\end{theoremAFA}

\bigskip

The intuition behind the theorem is simple. The proof uses the same construction in all three cases: if the observed value of $X_{k_n,n}$ is $x$, the estimator chooses the state $\theta$ that makes both $F(x|\theta)$ close to the expected percentile rank $(n-k_n+1)/(n+1)$ and $1-F(x|\theta)$ close to the complementary rank $k_n/(n+1)$.\footnote{As shown in Appendix~\ref{APP:PROOFS}, in each state $\theta$ the random variable $F(X_{k_n,n}|\theta)$ has a Beta$(n-k_n+1,k_n)$ distribution. Thus the expected percentile of $X_{k_n,n}$ is the mean of that distribution, $(n-k_n+1)/(n+1)$.} If both $k_n$ and $n-k_n$ grow unboundedly large, then in each state $\theta$ the order statistic is asymptotically equal to the corresponding population quantile: observing $X_{k_n,n}$ asymptotically reveals the value of the quantile $F^{-1}((n-k_n+1)/(n+1)|\theta)\) because, with both sides of the order statistic containing many observations, sampling noise around this quantile vanishes. Under \eqref{EQ:ID}, this quantile information is sufficient to recover the state consistently.

Cases (ii) and (iii) are different, because the same estimator is driven by tail information. If $n-k_n$ is bounded, $X_{k_n,n}$ is one of the few lowest observations in the sample. Such a statistic can identify the state only if the lower tail is arbitrarily informative, in the sense of \eqref{EQ:LI}. If instead $k_n$ is bounded, $X_{k_n,n}$ is one of the few highest observations, and consistency requires the upper-tail condition \eqref{EQ:UI}. With a bounded likelihood ratio in the relevant tail, a bounded number of extreme observations carries only a bounded amount of information, even as the sample size grows.

\bigskip\noindent\textbf{Examples.} In a normal location experiment, both \eqref{EQ:LI} and \eqref{EQ:UI} hold: likelihood ratios vanish in the lower tail, and grow unbounded in the upper tail. Therefore \emph{any} sequence of order statistics has AFA. The same is true for location families where the noise has bounded support, such as uniform noise. By contrast, with logistic noise, likelihood ratios are bounded away from both zero and infinity in both tails. Hence bounded upper-rank and bounded lower-rank order statistics fail to have AFA, but every middle order statistic is still asymptotically fully informative. Finally, with exponential noise, the lower bound of the signal distribution moves with the state, but the upper tail has only bounded likelihood ratios. Thus, \eqref{EQ:LI} holds and \eqref{EQ:UI} fails: all sequences except those with bounded upper rank have AFA.

\subsection{Characterization of Limit Experiments}
In this section we characterize the limit residual uncertainty that remains in the bounded-rank cases where AFA fails. For fixed $r\geq1$, we identify the limit experiment to which the $r$-th lowest and $r$-th highest order statistics $X_{n-r+1,n}$ and $X_{r,n}$ converge as the sample size $n$ grows large.

The unbounded lower informativeness condition \eqref{EQ:LI} requires sufficiently low realizations to separate every pair of states, in the sense that the likelihood ratio in favor of the higher state becomes arbitrarily small along the lower tail. Now suppose that \eqref{EQ:LI} holds for some pairs of states and not for others. Then we should expect $X_{n-r+1,n}$ to allow asymptotic learning of a \emph{set} of states---those not separated in the sense of \eqref{EQ:LI} from the true state---and at the same time provide some information \emph{within} that set. Symmetrically, if condition \eqref{EQ:UI} holds for some pairs of states and not for others, we should expect the analogous conclusion to hold for $X_{r,n}$.

To formalize this intuition, call two states $\theta'>\theta$ \emph{lower-nonseparable} if
\begin{equation}
\label{EQ:inf}
\inf\left\{ \frac{F(x|\theta')}{F(x|\theta)} \; : \; \underline{x}_\theta < x < \overline{x}_\theta \right\}>0.
\end{equation}
A \emph{lower class} is a maximal subset of states $L\subseteq\Theta$ such that every pair $\theta'>\theta$ in $L$ is lower-nonseparable. Note that, by \eqref{EQ:D} and \eqref{EQ:inf}, for every such pair we must have $\underline{x}_{\theta}=\underline{x}_{\theta'}$. In what follows we write $\underline{x}_L$ to denote this common lower bound. Note also that the lower classes form a partition of $\Theta$, and every state between two states in the same lower class also belongs to the class---in other words, each class is convex, meaning an interval when $\Theta$ is an interval, or a consecutive block of states when $\Theta$ is finite.\footnote{Formally, define a relation $\sim$ on $\Theta$ by letting $\theta\sim\vartheta$ if either $\theta=\vartheta$ or the pair $\max\{\theta,\vartheta\}>\min\{\theta,\vartheta\}$ is lower-nonseparable in the sense of \eqref{EQ:inf}. This relation is clearly reflexive and symmetric. It is also transitive, because three states $\theta_1<\theta_2<\theta_3$ with $\theta_1\sim\theta_2$ and $\theta_2\sim\theta_3$ must have the same lower support bound, while $F(x|\theta_3)/F(x|\theta_1)=[F(x|\theta_3)/F(x|\theta_2)][F(x|\theta_2)/F(x|\theta_1)]$ is bounded away from zero near that bound, hence $\theta_1\sim\theta_3$. For convexity, observe that, if $\theta_1<\theta_2<\theta_3$ and $\theta_1\sim\theta_3$, then, by \eqref{EQ:D}, $F(x|\theta_3)\leq F(x|\theta_2)$. Therefore, $F(x|\theta_2)/F(x|\theta_1)\geq F(x|\theta_3)/F(x|\theta_1)$, and since the right-hand side is bounded away from zero near the common bound, $\theta_2\sim\theta_1$ follows.}

Now, for every lower class $L$, fix once and for all a reference state $\theta_{\!L}\in L$, let
\[
W^L_{r,n} := \log\big(nF(X_{n-r+1,n}|\theta_L)\big),
\]
and define a strictly increasing function $\lambda(\cdot|L):L\to\R$ as follows: for all $\theta\in L$,\footnote{\label{FN:lambda}The function $\lambda(\cdot|L)$ is indeed well-defined and strictly increasing. If $\theta>\theta_L$ then $F(x|\theta)/F(x|\theta_L)$ is increasing by \eqref{EQ:MLR2}, bounded above by one by \eqref{EQ:D}, and bounded away from zero by \eqref{EQ:inf}. Thus, its limit as $x\to\underline{x}_L$ exists and is positive. Moreover, the limit is strictly below one---if it were one, the increasing ratio, being bounded above by one, would be identically one on the interior, contradicting \eqref{EQ:SD}. Therefore, $\lambda(\theta|L)$ is finite and strictly positive. If $\theta<\theta_L$, applying the same argument to the pair $\theta_L>\theta$ shows that $\lambda(\theta|L)$ is finite and strictly negative. To see that $\lambda(\cdot|L)$ is strictly increasing, note that for all $\theta'>\theta$ in $L$ we have $\lambda(\theta'|L)-\lambda(\theta|L)=\log\lim_{x\to\underline{x}_L}F(x|\theta)/F(x|\theta')>0$, where the inequality follows from \eqref{EQ:SD} and \eqref{EQ:MLR2}.}
\[
\lambda(\theta|L)
:=
\lim_{x\to\underline{x}_L}\big[R(x|\theta)-R(x|\theta_{\!L})\big].
\]
The following theorem shows that the limit information of the sequence $X_{n-r+1,n}$ is the information of a class-stratified location experiment. First, observing $X_{n-r+1,n}$ ensures asymptotic learning of the lower class containing the true state. Second, the observationally equivalent  $W^L_{r,n}$ provides the limit within-class information,\footnote{Since the function $x\mapsto\log(nF(x|\theta_L))$ is strictly increasing only for $\underline{x}_L<x<\overline{x}_{\theta_L}$, strictly speaking the random variable $W^L_{r,n}$ is not globally observationally equivalent to $X_{n-r+1,n}$. However, for every $\theta\in L$ the event $\underline{x}_L<X_{n-r+1,n}<\overline{x}_{\theta_L}$ has $P_\theta$-probability converging to one. Hence this lack of global observational equivalence is asymptotically irrelevant. Indeed, we could alternatively define $W^L_{r,n}$ by using the function $x\mapsto\log(nF(x|\theta_L)+a_nC(x))$ for some sequence $a_n\to0$ and bounded strictly increasing function $C:\R\to(0,1)$. With this alternative definition, $W^L_{r,n}$ would be a strictly increasing transformation of $X_{n-r+1,n}$, and part (ii) of Theorem~\ref{THM:CLASS1} would still be true.} in the form of an experiment with reparametrized state $\lambda(\theta|L)$ and additive noise distributed as the logarithm of a sum of exponential random variables.

\begin{theoremfour}
\label{THM:CLASS1}
Assume \eqref{EQ:ID}. Fix $r\geq1$.
\begin{enumerate}
\item[(i)]
There exists a sequence of functions $\phi_n:\R\to\Theta\,$ such that, for every lower class $L$ and every state $\theta\in L$,
\begin{equation*}
\inf_{\vartheta\in L}\Big|\phi_n(X_{n-r+1,n})-\vartheta\Big|
\tconv0.
\end{equation*}
\item[(ii)]
Let $E_1,\ldots,E_r$ be independent exponentials and $Z_r:=\log(E_1+\cdots+E_r)$. For every lower class $L$ and every state $\theta\in L$,
\[
W^L_{r,n}
\tconvd
\lambda(\theta|L)+Z_r.
\]
\end{enumerate}
\end{theoremfour}

The lower-class learning statement in the theorem, part (i), uses the same percentile-matching construction as in Theorem~\ref{THM:AFA}. For fixed $r$, the relevant percentile rank $r/(n+1)$ tends to zero. Lower-tail observations do not separate states within a lower class, so the estimator cannot be state-consistent, but it remains class-consistent: it identifies the lower class containing the true state. Part (ii) relies on a standard result in extreme value theory: under the true state $\theta$, the probability transform $F(X_{n-r+1,n}|\theta)$, scaled by $n$, converges in distribution to a sum of $r$ exponentials. Replacing $\theta$ by the reference state $\theta_L$ adds a factor $\exp(\lambda(\theta|L))$, and taking logs yields the location form in the statement.

\medskip

The arguments for the fixed upper-rank statistic $X_{r,n}$ mirror those for $X_{n-r+1,n}$. Call two states $\theta'>\theta$ \emph{upper-nonseparable} if
\begin{equation}
\label{EQ:sup}
\sup\left\{ \frac{1-F(x|\theta')}{1-F(x|\theta)} \; : \; \underline{x}_\theta < x < \overline{x}_\theta \right\}<\infty.
\end{equation}
An \emph{upper class} is a maximal subset of states $U\subseteq\Theta$ such that every pair $\theta'>\theta$ in $U$ is upper-nonseparable. By \eqref{EQ:D} and \eqref{EQ:sup}, for every such pair we must have $\overline{x}_{\theta}=\overline{x}_{\theta'}$, so we can write $\overline{x}_U$ to denote the common upper bound. Just like lower classes, upper classes form a partition of $\Theta$, and again each class is convex. For every upper class $U$, fixing a reference state $\theta_{U}\in U$ and defining
\[
Y^U_{r,n} := -\log\big(n\big[1-F(X_{r,n}|\theta_U)\big]\big),
\]
as well as the increasing function $\upsilon_U:U\to\R$ by\footnote{The argument for why $\upsilon(\cdot|U)$ is well defined and strictly increasing mirrors the one given in Footnote~\ref{FN:lambda} for $\lambda(\cdot|L)$.}
\[
\upsilon_U(\theta)
:=
\lim_{x\to\overline{x}_U}\big[H(x|\theta_U)-H(x|\theta)\big],
\]
we obtain the following.

\begin{theoremfourstar}
\label{THM:CLASS2}
Assume \eqref{EQ:ID}. Fix $r\geq1$.
\begin{enumerate}
\item[(i)]
There exists a sequence of functions $\phi_n:\R\to\Theta\,$ such that, for every upper class $U$ and every state $\theta\in U$,
\begin{equation*}
\inf_{\vartheta\in U}\Big|\phi_n(X_{r,n})-\vartheta\Big|
\tconv0.
\end{equation*}
\item[(ii)]
Let $E_1,\ldots,E_r$ be independent exponentials and $Z_r:=\log(E_1+\cdots+E_r)$. For every upper class $U$ and every state $\theta\in U$,
\[
Y^U_{r,n}
\tconvd
\upsilon(\theta|U)-Z_r.
\]
\end{enumerate}
\end{theoremfourstar}

\medskip\noindent\textbf{Maximum in Location Experiments.} Consider a location experiment $X_i=\theta+\ve_i$ and suppose that the noise hazard rate has a finite limit: for some $\alpha>0$, we have $f(\ve)/[1-F(\ve)]\to1/\alpha$ as $\ve\to\infty$. Then, for $\theta'>\theta$, we have $[1-F(x|\theta')]/[1-F(x|\theta)]\to\exp[(\theta'-\theta)/\alpha]$ as $x\to\infty$. Thus condition \eqref{EQ:UI} fails, there is only one upper class $U$, namely $\Theta$ itself, and $\upsilon(\theta|U)=(\theta-\theta_U)/\alpha$. For the maximum, $r=1$, Theorem~\ref{THM:CLASS2} therefore gives the limit experiment
\[
\upsilon(\theta|U)-Z_1=\frac{\theta-\theta_U}{\alpha}-\log E_1.
\]
Multiplying by $\alpha$ and adding $\theta_U$, this is the location experiment $\theta-\alpha\log E_1$, where $-\log E_1$ has the Gumbel distribution. This is exactly the non-full information limit in Theorem~3 of \citet{SSS}. If instead the hazard rate diverges, the upper-tail likelihood ratio becomes unbounded, upper classes are singletons, and the limit is fully informative, again matching the full-information case of that theorem.

\subsection{Comparison of Limit Experiments}

Theorems~\ref{THM:CLASS1} and \ref{THM:CLASS2} characterize the limit information of the statistics $X_{n-r+1,n}$ and $X_{r,n}$ for fixed $r$ as the sample size $n$ grows large. Once the lower class $L$ has been learned, the residual within-class information is described by $\lambda(\theta|L)+Z_r$. Taking $L$ as the state space, and viewing $\lambda(\theta|L)+Z_r$ as an experiment on this restricted state space, in this section we compare the informativeness of this experiment as $r$ varies. In other words, we ask: is the $r$-th lowest observation in a very large sample more or less accurate than the $(r+1)$-th lowest? Similarly, taking an upper class $U$ and viewing $\upsilon(\theta|U)-Z_r$ as an experiment on $U$, we ask whether the $r$-th highest is more or less accurate than the $(r+1)$-th highest.

The following theorem shows that accuracy increases with $r$, that is, more central extreme order statistics are more informative. In fact, increasing $r$ makes the limit experiment Blackwell-more informative. Moreover, connecting back to part (i) of Theorem~\ref{THM:AFA}, information becomes full as, after taking $n$ to infinity, we take $r$ to infinity.

\begin{theoremfive}
\label{THM:BLACKWELL1}
Fix a lower class $L$. For every $r\geq1$, the experiment $\{\lambda(\theta|L)+Z_r:\theta\in L\}$ becomes Blackwell-more informative as $r$ increases. Moreover, it is fully informative in the limit as $r\to\infty$. In particular, letting $\phi_r(x):=x-\log r$, for every $\theta\in L$ we have
\[
\phi_r(\lambda(\theta|L) + Z_r)
=
\lambda(\theta|L) + Z_r -\log r
\;\tconv\;
\lambda(\theta|L)
\qquad\text{as $r\to\infty$}.
\]
\end{theoremfive}

\begin{theoremfivestar}
\label{THM:BLACKWELL2}
Fix an upper class $U$. For every $r\geq1$, the experiment $\{\upsilon(\theta|U)-Z_r:\theta\in U\}$ becomes Blackwell-more informative as $r$ increases. Moreover, it is fully informative in the limit as $r\to\infty$. In particular, letting $\phi_r(x):=x+\log r$, for every $\theta\in U$ we have
\[
\phi_r(\upsilon(\theta|U) - Z_r)
=
\upsilon(\theta|U) - Z_r +\log r
\;\tconv\;
\upsilon(\theta|U)
\qquad\text{as $r\to\infty$}.
\]
\end{theoremfivestar}

The comparison in the first part of the theorem may seem surprising at first. The limit noise for rank $r+1$ involves the sum of one more independent exponential random variable than the limit noise for rank $r$. This implies that $E_1+\cdots+E_{r}$ is less dispersed than $E_1+\cdots+E_{r+1}$. However, the relevant noise terms are not the sums themselves, but their logarithms $Z_r$ and $Z_{r+1}$, and taking logs reverses the dispersive-order comparison in this particular case.\footnote{The simple proof of the theorem, in Appendix~\ref{APP:PROOF-BLACKWELL}, exploits the beta-gamma identity to obtain $\log(E_1+\cdots+E_{r})=\log(E_1+\cdots+E_{r+1})+\log B_r$, where $B_r$ has a Beta distribution and is independent of $E_1,\ldots,E_{r+1}$, so the limit experiment for rank $r$ is obtained from that for rank $r+1$ by adding independent noise.} The reason is that the additional exponential term increases the sum, but it also dampens fluctuations. Thus, increasing the fixed lower rank makes the residual noise less dispersed and the limit experiment more accurate. Indeed, the theorem shows that it does more than that: it makes the experiment for rank $r+1$ Blackwell-more informative.

\section{Multidimensional Experiments\label{SEC:MULTI}}

The results obtained so far concern single order statistics. Many selected-data scenarios, however, naturally lead to multidimensional experiments. A researcher or platform may disclose a block of observations remaining after eliminating some extreme draws on either side. An auction may reveal a set of winning or losing bids rather than a single price. A committee may be selected from a larger pool, as under peremptory challenge, where one side removes members with high signals and the other removes members with low signals. The relevant experiment in such cases is not a single statistic $X_{k,n}$, but rather a vector of consecutive order statistics. Given a block size $m$, let
\[
\mathbf{X}^m_{k,n}:=
\big(X_{k+m-1,n},\ldots,X_{k,n}\big).
\]
denote the block of $m$ consecutive order statistics (with $k+m-1\leq n$). As in the unidimensional case, more maximal selection means increasing the sample size $n$ while keeping the upper rank $k$ fixed. More minimal selection means increasing both $n$ and $k$ by the same number. The block size $m$ is kept fixed in both comparisons.

\subsection{Minimal and Maximal Block-Selection}
To compare vectors of order statistics, we adopt \citeapos{SSS} multidimensional accuracy notion. We detail the notationally easier case of more maximal selection---the case of more minimal selection is symmetric. Fix an upper rank $k$ and a block size $m$. For every $n\geq k+m-1$ and every state $\theta$, define the mapping $\varphi(\theta,\cdot):\R^m\to\R^m$ as follows: $\varphi(\theta,x_1,\ldots,x_m)=(y_1,\ldots,y_m)$, where $y_1,\ldots,y_m$ are defined recursively by
\begin{equation}
\label{EQ:multi1}
\Pr\nolimits_\theta
\big(
X_{k,n+1}\leq y_1
\big)
=
\Pr\nolimits_\theta
\big(
X_{k,n}\leq x_1
\big)
\end{equation}
and, for $\ell=1,\ldots,m-1$,
\begin{equation}
\label{EQ:multi2}
\Pr\nolimits_\theta
\big(
X_{k+\ell,n+1}\leq y_{\ell+1}
\,\big\vert\,
X_{k+\ell-1,n+1}=y_{\ell}
\big)
=
\Pr\nolimits_\theta
\big(
X_{k+\ell,n}\leq x_{\ell+1}
\,\big\vert\,
X_{k+\ell-1,n}=x_{\ell}
\big).
\end{equation}
The more-maximally-selected experiment $\mathbf{X}^m_{k,n+1}$ is more accurate than $\mathbf{X}^m_{k,n}$ if $\varphi(\theta,x_1,\ldots,x_m)$ is increasing in $\theta$ for every $(x_1,\ldots,x_m)$. Intuitively, after matching the false-positives induced by each coordinate conditional on the previous ones, the more accurate block also gives lower false-negatives. The definition for more minimal selection increasing accuracy is analogous.\footnote{The definition of multidimensional accuracy is order-dependent because the mapping $\varphi$ is defined recursively. The order used above is the one that makes the conditional distributions of successive coordinates right-truncated. For more minimal selection, we use the reverse ordering of the same block, so that successive conditional distributions are left-truncated. Since different orderings of a block are mere relabelings of the same observed vector, this order dependence concerns the verification of accuracy, not the information contained in the block itself.}

The following result gives the block counterpart of Theorems~\ref{THM:ACCORDION1} and~\ref{THM:ACCORDION2}. The theorem presents sufficient conditions. For $m=1$, they imply the sharp conditions in Theorems~\ref{THM:ACCORDION1} and~\ref{THM:ACCORDION2}, but they are stronger: as we explain below, observing a block involves controlling the behavior of the distribution $F$ after truncations, not only at the bounds of its support.

\begin{theoremsix}
\label{THM:MULTI}
If the hazard rate $h(x|\theta):=f(x|\theta)/[1-F(x|\theta)]$ is log-supermodular, then more minimal selection increases accuracy: $\mathbf{X}^m_{k+1,n+1}$ is more accurate than $\mathbf{X}^m_{k,n}$ for all $k,n,m$. If the reverse hazard rate $r(x|\theta):=f(x|\theta)/F(x|\theta)$ is log-supermodular, then more maximal selection increases accuracy: $\mathbf{X}^m_{k,n+1}$ is more accurate than $\mathbf{X}^m_{k,n}$ for all $k,m,n$.
\end{theoremsix}

To get an intuition, consider more maximal selection. Conditional on $X_{k,n}=x_1$, the next coordinate $X_{k+1,n}$ is the maximum of $n-k$ draws from the truncated distribution $F(\cdot|\theta)/F(x_1|\theta)$. Similarly, conditional on $X_{k+1,n}=x_2$, the next coordinate $X_{k+2,n}$ is the maximum of $n-k-1$ draws from $F(\cdot|\theta)/F(x_2|\theta)$, and so on. Thus, the relevant objects for more maximal selection are cumulative reverse hazard \emph{differences}, that is, objects of the form
\[
-\log\left[\frac{F(x_{\ell+1}|\theta)}{F(x_{\ell}|\theta)}\right]^{n-k-\ell+1}
\;\propto\;\;
R(x_{\ell+1}|\theta)-R(x_{\ell}|\theta)
\;=\;
\int_{x_{\ell+1}}^{x_{\ell}} r(t|\theta)\,dt.
\]
Log-supermodularity of the reverse hazard rate $r$ is the appropriate sufficient condition for controlling these differences---guaranteeing that the state comparisons dictated by log-supermodularity hold not only for the tail, but also for every truncated interval. The more-minimal-selection comparison is symmetric. More minimal selection involves truncation from below, and the relevant objects are the cumulative hazard differences, controlled by the hazard rate $h$.\footnote{The one-dimensional results in Theorems~\ref{THM:ACCORDION1} and~\ref{THM:ACCORDION2} can be viewed as boundary cases of the same argument. With a single order statistic, there is no previously observed coordinate, so there is ``truncation'' at the lower or upper bound of the support. Indeed, $H(x|\theta)=H(x|\theta)-H(\underline x_\theta|\theta)$ and $R(x|\theta)=R(x|\theta)-R(\overline x_\theta|\theta)$. For blocks of observations, the bounds are replaced by endogenous truncation points, and this is why the stronger rate conditions are needed.}

\subsection{Block Comparisons in Large Samples}

The large-sample results have immediate multidimensional counterparts. Here, as before, in large samples it is natural to allow the top rank $k$, and this time also the block size $m$, to depend on sample size. Thus, let $k_n$ and $m_n$ be sequences such that $k_n+m_n-1\leq n$, and consider the corresponding sequence of blocks $\mathbf{X}^{m_n}_{k_n,n}$. If the sequence contains a middle order statistic, that is, if for some sequence $j_n$ with $k_n\leq j_n\leq k_n+m_n-1$ we have both $j_n\to\infty$ and $n-j_n\to\infty$, then under \eqref{EQ:ID} the block asymptotically identifies the state. Indeed, the sequence $X_{j_n,n}$ already does so alone, and observing the whole block can only add information. If instead the block remains either entirely in the lower tail ($n-k_n$ is bounded) or entirely in the upper tail ($k_n+m_n$ is bounded), then the same tail conditions as in parts (ii) and (iii) of Theorem~\ref{THM:AFA} are necessary and sufficient for full learning.

When these conditions fail, the residual limit experiment is multidimensional. For example, fix $m\geq1$ and consider the block of the $m$ lower order statistics, $(X_{n,n},X_{n-1,n},\ldots,X_{n-m+1,n})$. By Theorem~\ref{THM:CLASS1}, given any lower class $L$ and state $\theta\in L$ we have
\[
\left(W^L_{1,n},\ldots,W^L_{m,n}\right)
\tconvd
\big(
\lambda(\theta|L)+Z_1,
\ldots,
\lambda(\theta|L)+Z_m
\big).
\]
Similarly, by Theorem~\ref{THM:CLASS2}, for every upper class $U$ and state $\theta\in U$,
\[
\left(Y^U_{1,n},\ldots,Y^U_{m,n}\right)
\tconvd
\big(\upsilon(\theta|U)-Z_1,\ldots,\upsilon(\theta|U)-Z_m\big).
\]
Thus, blocks refine the residual information characterized in Theorems~\ref{THM:CLASS1} and \ref{THM:CLASS2}, but they do not change the underlying class structure. The class is learned and, within the class, the limit is a multivariate location experiment with noise distributed as a dependent vector of logarithms of sums of exponentials.

\section{Information Aggregation in Auctions\label{SEC:AUCTIONS}}

Our theorems directly contribute to the literature on information aggregation in auctions pioneered by \citet{Wilson77} and \citet{Milgrom79}. Consider a symmetric auction with $n$ unit-demand bidders and $k<n$ identical objects. The auction is either discriminatory---each of the $k$ highest bidders receives an object at a price equal to the submitted bid---or uniform-price---each of the $k$ highest bidders receives an object for a price equal to the highest rejected bid.

Bidder $i$ observes a private signal $X_i$, and conditional on a parameter $\theta$ the signals are i.i.d.~with distribution $F(\cdot|\theta)$. Values are symmetric and affiliated, as in \citet{Milgrom-Weber,Milgrom-Weber2}, so in the standard symmetric equilibria bids are strictly increasing functions of signals. The order of bids then coincides with the order of signals. In particular, the lowest winning bid is informationally equivalent to the $k$-th highest signal $X_{k,n}$. Similarly, the highest rejected bid is informationally equivalent to $X_{k+1,n}$. Thus, information comparisons based on the price---intended as either the lowest winning bid or the highest losing bid---are a direct application of our order statistics comparisons. Asking whether more competition improves information aggregation is the same as asking how the accuracy of these order statistics changes when one more bidder is added.

If the number of objects is held fixed, we compare $X_{k,n}$ with $X_{k,n+1}$, a more-maximal-selection comparison. If instead there is also an additional object, we compare $X_{k,n}$ with $X_{k+1,n+1}$, a more-minimal-selection comparison. Thus, Theorems~\ref{THM:ACCORDION1} and \ref{THM:ACCORDION2} translate immediately into the comparative statics stated in Proposition~\ref{PRP:AUCTIONS1} below. This is the finite-auction counterpart of the information-aggregation results of \citet{Wilson77} and \citet{Milgrom79}, applying away from the perfectly competitive limit, and covering both discriminatory and uniform-price formats.

\begin{proposition}
\label{PRP:AUCTIONS1}
The signal distribution has log-supermodular cumulative reverse hazard $R(\cdot|\theta)=-\log F(\cdot|\theta)$ if and only if, given any number of bidders and objects, holding fixed the number of objects, the informativeness of the equilibrium price increases in the number of bidders. The cumulative hazard $H(\cdot|\theta)=-\log(1-F(\cdot|\theta))$ is log-supermodular if and only if, given any number of bidders and objects, the informativeness of the equilibrium price increases as the number of bidders and the number of objects are both increased by the same number.
\end{proposition}

Thus, under the two log-supermodularity assumptions, soliciting an additional bidder can only improve the information contained in the price, provided the number of objects either stays fixed or increases by one. Holding the number of objects fixed pushes the relevant price statistic further into the upper tail: the price remains the same-rank order statistic but is selected from a larger pool of bidders. This comparison is governed by the cumulative reverse hazard $R(x|\theta)=-\log F(x|\theta)$. Increasing both bidders and objects by one instead keeps the number of rejected bidders fixed, and pushes the relevant statistic toward the lower tail; this comparison is governed by the cumulative hazard $H(x|\theta)=-\log(1-F(x|\theta))$. Thus the two branches of the accordion correspond exactly to the two natural ways one can think about increasing market size.

Note that the proposition above compares informativeness of equilibrium prices, not their levels. Competition may raise or lower expected prices, depending on the auction format and the exact assumed environment. What matters here is that, because equilibrium bids are strictly increasing in signals, the price is a monotone transformation of an order statistic. Therefore any accuracy comparison among order statistics carries over to the price experiment. In particular, with signal distributions such as normal or logistic location, for which both $H$ and $R$ are log-supermodular, price information improves along both possible paths.

The next proposition, a direct implication of Theorem~\ref{THM:AFA}, gives the perfectly competitive limit. If both the number of objects and the number of losers grow without bound, the price is an intermediate order statistic. Full information aggregation follows from \eqref{EQ:MLR} and the identification condition \eqref{EQ:ID}. If the number of objects remains bounded, the winning price, or the highest rejected bid in the uniform-price auction, is one of the top order statistics, and full information aggregation requires \eqref{EQ:UI}. Symmetrically, if the number of losers $n-k_n$ remains bounded, full information aggregation requires \eqref{EQ:LI}. 

\begin{proposition}
\label{PRP:AUCTIONS2}
A sequence of auctions satisfies full information aggregation if and only if both \eqref{EQ:ID} and one of the following conditions hold:
\begin{enumerate}
\item[(i)] Both the number of objects $k_n$ and the number of losers $n-k_n$ grow unbounded.
\item[(ii)] The number of objects $k_n$ is a bounded sequence and the signals satisfy \eqref{EQ:UI}.
\item[(iii)] The number of losers $n-k_n$ is a bounded sequence and the signals satisfy \eqref{EQ:LI}.
\end{enumerate}
\end{proposition}

Part (ii) of the proposition recovers the logic of \citeapos{Milgrom79} unbounded-informativeness condition: with only finitely many winning signals, the price fully reveals the state only if sufficiently high realizations make the likelihood ratio arbitrarily large. More generally, the proposition unifies the standard full-information aggregation theorem with its lower-tail counterpart and with the intermediate case, where full aggregation is automatic under \eqref{EQ:ID}.

Proposition~\ref{PRP:AUCTIONS2} is also closely related to \citet{Pesendorfer-Swinkels-1997}. In their pure common-value model, the equilibrium price in a $k$-unit uniform-price auction aggregates information if and only if the sequence of auctions is ``doubly large'': both the number of objects $k_n$ and the number of losing bidders $n-k_n$ diverge. Conditional on pivotality of a bid, the winner's curse and the loser's curse convey information about the common value.\footnote{\citet{Pesendorfer-Swinkels-2000} consider an environment with both a common quality component and idiosyncratic tastes, where information aggregation must be reconciled with allocative efficiency.} As Proposition~\ref{PRP:AUCTIONS2} shows, \citeapos{Pesendorfer-Swinkels-1997} result critically hinges on failure of both \eqref{EQ:LI} and \eqref{EQ:UI}. When both $k_n$ and $n-k_n$ diverge, the price is a middle order statistic: full information aggregation obtains by \eqref{EQ:MLR} and \eqref{EQ:ID} alone, matching the double-largeness intuition. However, when one side of the market remains bounded, the price is an extreme order statistic, and full aggregation requires unbounded informativeness: \eqref{EQ:UI} if the number of objects is bounded, \eqref{EQ:LI} if the number of losers is bounded. Thus our result recovers the double-largeness condition as the interior case, while also characterizing precisely what additional tail conditions are needed at the two boundaries.

\section{Information Aggregation in Voting\label{SEC:VOTING}}

We now turn to voting, an application developed in substantially more detail in the companion paper \citet{DOS-voting}. Consider a committee or jury of size $n$ facing a binary decision, say acquittal versus conviction. Juror $i$ observes a private signal $X_i$, and conditional on the state $\theta$ the signals are i.i.d.~with distribution $F(\cdot|\theta)$. We consider the following \emph{rule $k$} for voting: the jury convicts if and only if at least $k$ jurors vote for conviction. This is the canonical environment studied in the strategic-voting literature on juries, including \citet{Feddersen-Pesendorfer} in a binary-signal model and \citet{Duggan-Martinelli} in a continuous-signal model.

The key insight in \citet{DOS-voting} is that strategic voting turns the jury problem into an order-statistic experiment. If all jurors use cutoff $x$, then the jury convicts if and only if at least $k$ signals exceed $x$, equivalently if and only if $X_{k,n}\geq x$. Thus, the decisive signal is $X_{k,n}$, the $k$-th highest signal. Under strategic voting the cutoff is the signal at which an individual juror, conditioning on being pivotal, is indifferent between the two alternatives. Thus, the cutoff is such that a juror with signal $x$, conditional on exactly $k-1$ other jurors voting for conviction and $n-k$ voting for acquittal, is indifferent. Therefore, the jury decision coincides with the decision that the juror with signal $X_{k,n}$ would take if she \emph{knew} that her signal is the $k$-th highest.

\begin{proposition}
\label{PRP:VOTING1}
The symmetric equilibrium in a jury of size $n+1$ with rule $k+1$ (resp.~rule $k$) achieves a more accurate decision than the symmetric equilibrium in the jury of size $n$ with rule $k$, for every $n$ and $k$, if and only if the signals' cumulative hazard (resp.~cumulative reverse hazard) is log-supermodular.\end{proposition}

The proposition compares two natural ways of expanding a jury. Starting from size $n$ with rule $k$, adding one juror while raising the quota to $k+1$ compares $X_{k,n}$ with $X_{k+1,n+1}$. This is the more-minimal-selection comparison, governed by the cumulative hazard $H$. Adding one juror while keeping the quota fixed at $k$ compares $X_{k,n}$ with $X_{k,n+1}$. This is the more-maximal-selection comparison, governed by the cumulative reverse hazard $R$. Thus the two branches of the accordion correspond to the two basic ways of expanding a voting rule: keeping fixed the number of votes needed for conviction, or keeping fixed the number of votes needed for acquittal.

Considering a sequence of voting rules $k_n$ and letting the number of voters become large, our Theorem~\ref{THM:AFA} delivers the Condorcet-type implication stated in Proposition~\ref{PRP:VOTING2} below. A proportional-quota rule is one for which both the number of votes required for conviction ($k_n$) and the number of votes required for acquittal ($n-k_n$) grow without bound. Equivalently, the pivotal signal $X_{k,n}$ is an intermediate order statistic. The pivotal voter then conditions on many signals above and many signals below her own, sampling noise around the corresponding quantile vanishes, and under \eqref{EQ:ID} the state is asymptotically learned---the probability of a correct decision converges to one.\footnote{The classic Condorcet result with sincere voting is different. In a two-state model with innocence $\theta_0$ and guilt $\theta_1>\theta_0$, suppose juror $i$ votes to convict if and only if $X_i\geq c$. Then a conviction vote has probability $1-F(c|\theta)$ in state $\theta$. If $k_n/n\to q\in(0,1)$ and $1-F(c|\theta_0)<q<1-F(c|\theta_1)$, the law of large numbers (applied to the fraction of conviction votes) implies that the probability of conviction converges to zero in state $\theta_0$ and to one in state $\theta_1$. Strategic voting changes the cutoff, because voters condition on pivotality.} By contrast, with a finite quota for acquittal, e.g.~with unanimity for conviction, a bounded number of low signals is enough to acquit, and the pivotal signal is a lower-order statistic. Asymptotic correctness then requires condition \eqref{EQ:LI}. Symmetrically, if the rule has a finite quota for conviction, a bounded number of high signals is enough to convict, and asymptotic correctness requires \eqref{EQ:UI}.

\begin{proposition}
\label{PRP:VOTING2}
Assume signals satisfy \eqref{EQ:ID}. The probability of a correct decision converges to one if and only if (i) the rule is a proportional-quota rule, or (ii) the rule is a finite-quota rule for acquittal (resp.~conviction), and signals satisfy \eqref{EQ:LI} (resp.~\eqref{EQ:UI}). 
\end{proposition}

Our proposition recasts the lessons from the existing literature in order-statistic terms: whether large juries aggregate information is governed by whether the relevant pivotal order statistic asymptotically separates the states. \citeapos{Feddersen-Pesendorfer} insight is that strategic voting changes the meaning of a unanimous conviction: a juror who is pivotal under unanimity conditions on all other jurors voting to convict, and this event may overwhelm her private signal. This explains why unanimity can perform poorly under strategic voting. But this is not a general problem with unanimity. \citet{Duggan-Martinelli} show that, in a continuous-signal model, unanimity is asymptotically correct precisely when arbitrarily strong signals of innocence are available. In our terminology, unanimity is a finite-quota rule for acquittal, so the pivotal statistic is a lower extreme, and asymptotic correctness requires \eqref{EQ:LI}.\footnote{The companion paper \citet{DOS-voting} further shows that \citeapos{Feddersen-Pesendorfer} bad unanimity outcome relies on unstable equilibria.} In other words, the general lesson is not that unanimity is an inferior voting rule, but rather that unanimity puts all informational weight on the lower tail of the signals' distribution.

\medskip
\noindent\textbf{Fixed-Jury-Size Comparisons of Voting Rules.}
The companion paper \citet{DOS-voting} develops the voting application further in two directions that we leave outside the present paper. First, a complete theory of strategic voting must address equilibrium selection and stability. The order-statistic observation identifies the statistic on which a symmetric cutoff equilibrium is based, but it does not by itself say whether the relevant equilibrium is stable or whether other equilibria are more compelling. This issue is especially important for interpreting the unanimity result in \citet{Feddersen-Pesendorfer}. In this regard, the companion paper shows that the bad-unanimity outcome hinges on unstable equilibria.

Second, a full comparison of voting rules for a fixed jury size $n$ requires \emph{horizontal comparisons} in the accordion. Holding $n$ fixed and changing the rule from $k$ to $k+1$ compares $X_{k,n}$ with $X_{k+1,n}$. Such adjacent order statistics are generally not comparable in the global Lehmann accuracy order. Nevertheless, for voting the relevant question is often local: given the equilibrium cutoff, does one rule generate a better error trade-off than another in a neighborhood of the decision actually implemented by the equilibrium? Answering this question requires a \emph{local} version of accuracy, tailored to binary decisions. Developing that notion, and using it to compare different rules for a fixed jury size, is a relevant but separate exercise. For this reason, here we focus on the vertical accordion comparisons that follow directly from Theorems~\ref{THM:ACCORDION1} and~\ref{THM:ACCORDION2}, and on the large-jury Condorcet-type implication of Theorem~\ref{THM:AFA}. The companion paper provides the fuller local and equilibrium analysis.

\appendix

\section{Proofs\label{APP:PROOFS}}

For all $k$ and $n$, write $F_{k,n}(x|\theta)$ to denote the cumulative distribution function of $X_{k,n}$. Thus,
\[
F_{k,n}(x|\theta)
:=
\Pr\nolimits_\theta(X_{k,n}\leq x)
=
\sum^{n}_{i=n-k+1}\binom{n}{i}F^{i}(x|\theta)[1-F(x|\theta)]^{n-i}.
\]

\subsection{Proof of Theorem~\ref{THM:ACCORDION1}}
\label{APP:PROOF-THM1}

\noindent\emph{Necessity}. Assume that more minimal selection increases accuracy. In particular, for every pair $m>n$, the minimum of $m$ draws is more accurate than the minimum of $n$ draws, that is, $X_{m,m}$ is more accurate than $X_{n,n}$. The distribution function of $X_{n,n}$ is
\[
F_{n,n}(x|\theta)=1-\big[1-F(x|\theta)\big]^n=1-\exp\big[-nH(x|\theta)\big]
\]
Therefore, the accuracy comparison between $X_{m,m}$ and $X_{n,n}$ says that $H^{-1}((n/m)H(x|\theta)|\theta)$ is increasing in $\theta$ for every $x$. Equivalently, for every pair $\theta'>\theta$, defining $y$ by $H(y|\theta)=(n/m)H(x|\theta)$,
\[
\frac{H(y|\theta')}{H(y|\theta)}
\leq
\frac{H(x|\theta')}{H(x|\theta)}.
\]
Since $n/m<1$, we have $y<x$. Thus, the desired monotonicity of the ratio $H(\cdot|\theta')/H(\cdot|\theta)$ holds whenever $H(y|\theta)/H(x|\theta)$ is a rational number. By continuity, the same inequality holds for all $y<x$. Indeed, take an arbitrary $y<x$, let $a=H(y|\theta)/H(x|\theta)$, and take a sequence of rational numbers $a_i\to a$. Defining $y_i$ by $H(y_i|\theta)=a_iH(x|\theta)$, the inequality holds for $y_i$, and since $y_i\to y$ as $a_i\to a$ by continuity of $H$, it also holds for $y$. We have thus shown that $H$ is log-supermodular.

\medskip
\noindent\emph{Sufficiency}. Assume that $H$ is log-supermodular and fix any $k\geq1$ and $n\geq k$. We have to show
\begin{equation}\label{eq:Lehmann}
F_{k+1,n+1}^{-1}(F_{k,n}(x\mid\theta)\mid\theta)
\quad
\text{is increasing in $\theta$ for every $x$.}
\end{equation}

In each state $\theta$ the random variable $H(X_{k,n}|\theta)$ has the same distribution as the $k$-th highest order statistic in a sample of $n$ independent exponentials:
\[
\Pr\nolimits_\theta(H(X_{k,n}|\theta)\leq x)
=
\sum^{n}_{r=n-k+1}\binom{n}{r}(1-e^{-x})^{r}(e^{-x})^{n-r}
=:
G_{k,n}(x).
\]
Thus, $F_{k,n}(x|\theta)=G_{k,n}(H(x|\theta))$ and $F_{k+1,n+1}(x|\theta)=G_{k+1,n+1}(H(x|\theta))$. Let
\[
\eta(\cdot):=G^{-1}_{k+1,n+1}(G_{k,n}(\cdot)).
\]
Then we can write \eqref{eq:Lehmann} as follows:
\begin{equation}\label{eq:Lehmann2}
H^{-1}\big(\eta(H(x|\theta))\big|\theta\big)
\quad
\text{is increasing in $\theta$ for every $x$.}
\end{equation}

Fix any two states $\theta$ and $\theta'>\theta$ now, and define $T(\cdot):=H(H^{-1}(\cdot|\theta)|\theta')$.
Since $H^{-1}(\cdot|\theta)$ is increasing, to prove \eqref{eq:Lehmann2} it suffices to show that $\eta(T(H(x|\theta)))
\geq T(\eta(H(x|\theta)))$ for every $x$. Equivalently, we prove:
\begin{equation}\label{eq:Lehmann3}
\eta(T(y))
\geq
T(\eta(y))
\quad
\text{for every $y$.}
\end{equation}

To this end, we first show that the $(k+1)$-th highest of $n+1$ independent exponentials is smaller than the $k$-th highest of $n$ in the star-order, that is, we have:
\begin{lemma}\label{lem:exp}
The mapping $y\mapsto\eta(y)/y$ is decreasing.
\end{lemma}
\begin{proof}
Let $E_1,\ldots,E_{n+1}$ be independent exponentials, and let $j=n-k+1$. The cumulative distribution function of the $j$-th lowest order statistic is $G_{k+1,n+1}$. For every $a>0$, consider the sample $E_1,\ldots,E_n,E_a$, where $E_a=E_{n+1}/a$, and write $G_a$ for the cumulative distribution function of the corresponding $j$-th lowest order statistic. By Theorem~2 in \citet{Yu}, the $j$-th lowest order statistic in the homogeneous sample $E_1,\ldots,E_n,E_{n+1}$ is star-smaller than in the heterogeneous sample $E_1,\ldots,E_n,E_a$, that is,
\begin{equation}\label{eq:GG}
y\mapsto\frac{G^{-1}_a(G_{k+1,n+1}(y))}{y}
\quad
\text{is increasing for every $a>0$.}
\end{equation}

Since $E_a$ diverges almost surely as $a\to0$, with probability one $E_a>\max\{E_1,\ldots,E_n\}$ for every sufficiently small $a$. Thus, with probability one, for every sufficiently small $a$ the $j$-th lowest order statistic of $E_1,\ldots,E_n,E_a$ coincides with the $j$-th lowest order statistic of $E_1,\ldots,E_n$. It follows that the $j$-th lowest order statistic of $E_1,\ldots,E_n,E_a$ converges almost surely, and hence in distribution, to the $j$-th lowest order statistic of $E_1,\ldots,E_n$. In other words, $G_a(x)\to G_{k,n}(x)$ for every continuity point $x$ of $G_{k,n}$, and since $G_{k,n}$ is continuous on $(0,\infty)$, this convergence holds for every $x>0$. Moreover, since $G_a$ and $G_{k,n}$ are continuous and strictly increasing on $(0,\infty)$, their inverse functions satisfy $G_a^{-1}(u)\to G^{-1}_{k,n}(u)$ for every $u\in(0,1)$, and hence
\[
G^{-1}_a(G_{k+1,n+1}(y))\to G^{-1}_{k,n}(G_{k+1,n+1}(y))=\eta^{-1}(y)
\qquad
\text{for every $y>0$}.
\]
Being the pointwise limit of the increasing mapping in \eqref{eq:GG}, the mapping $y\mapsto\eta^{-1}(y)/y$ is also increasing. This implies the result. Indeed, since $\eta$ is increasing, for all $y_2>y_1>0$ we have $\eta(y_2)\geq\eta(y_1)$. Thus, $\eta^{-1}(\eta(y_2))/\eta(y_2)\geq\eta^{-1}(\eta(y_1))/\eta(y_1)$, that is, $\eta(y_2)/y_2\leq\eta(y_1)/y_1$.
\end{proof}

\bigskip

Next, we record two properties of the mapping $T$.
\begin{lemma}\label{lem:T}
The mapping $T$ satisfies $T(y)\leq y$ for every $y$. Moreover, $T(y)/y$ is increasing.
\end{lemma}
\begin{proof}
By \eqref{EQ:D} we have $H(x|\theta')\leq H(x|\theta)$ for all $x$, and hence
\[
T(y)=H(H^{-1}(y|\theta)|\theta')\leq H(H^{-1}(y|\theta)|\theta)=y
\qquad
\text{for all $y$}.
\]
To prove the second claim, observe that
\[
\log T(y)-\log y
=
\log H\big(H^{-1}(y|\theta)\big|\theta'\big)-\log H\big(H^{-1}(y|\theta)\big|\theta\big)
\]
is increasing in $y$, because $H$ is log-supermodular and $H^{-1}(\cdot|\theta)$ is increasing.
\end{proof}

\bigskip

Finally, since the $k$-th highest of $n$ independent exponentials first-order stochastically dominates the $(k+1)$-th highest of $n+1$, we have $G_{k+1,n+1}(\cdot)\geq G_{k,n}(\cdot)$ and hence $\eta(y)\leq y$. Thus, by the second claim in Lemma~\ref{lem:T},
\begin{equation}\label{eq:T1}
 \frac{T(y)}{y}
 \geq
 \frac{T(\eta(y))}{\eta(y)},
\end{equation}
while by Lemma~\ref{lem:exp} and the first claim in Lemma~\ref{lem:T},
\begin{equation}\label{eq:T2}
 \frac{\eta(T(y))}{T(y)}
 \geq
 \frac{\eta(y)}{y}.
\end{equation}
Combining \eqref{eq:T1} and \eqref{eq:T2} gives the desired inequality \eqref{eq:Lehmann3}.

\subsection{Proof of Theorem~\ref{THM:EXP1}}
\label{APP:PROOF-THM2}

The implication from (ii) to (i) is trivial. To prove that (iii) implies (ii), let $\theta\mapsto\lambda(\theta)$ and $x\mapsto \xi(x)$ be strictly increasing functions such that, in every state $\theta$, the random variable $\xi(X_i)-\lambda(\theta)$ has the exponential distribution. Since $\xi$ is strictly increasing, $\xi(X_{k,n})-\lambda(\theta)$ is distributed as the $k$-th highest of $n$ independent exponentials, so that $F_{k,n}(x|\theta)=G_{k,n}(\xi(x)-\lambda(\theta))$. We have to show that, for all $k\geq1$ and $n\geq k$, both $F_{k+1,n+1}^{-1}(F_{k,n}(x|\theta)|\theta)$ and $F_{k,n}^{-1}(F_{k,n+1}(x|\theta)|\theta)$ are increasing in $\theta$ for every $x$, that is, both
\begin{equation}
\label{EQ:GG}
\xi^{-1}\big(G_{k+1,n+1}^{-1}(G_{k,n}(\xi(x)-\lambda(\theta))+\lambda(\theta)\big)
\quad\text{and}\quad
\xi^{-1}\big(G_{k,n}^{-1}(G_{k,n+1}(\xi(x)-\lambda(\theta))+\lambda(\theta)\big)
\end{equation}
are increasing in $\theta$ for every $x$. \citet{Khaledi-Kochar} prove that $G_{k+1,n+1}$ is less dispersed than $G_{k,n}$ and $G_{k,n}$ is less dispersed than $G_{k,n+1}$, that is, both $G^{-1}_{k+1,n+1}(u)-G^{-1}_{k,n}(u)$ and $G^{-1}_{k,n}(u)-G^{-1}_{k,n+1}(u)$ are decreasing in $u$. This implies that both
\[
G_{k+1,n+1}^{-1}(G_{k,n}(\xi(x)-\lambda(\theta)))+\lambda(\theta)
\quad\text{and}\quad
G_{k,n}^{-1}(G_{k,n+1}(\xi(x)-\lambda(\theta)))+\lambda(\theta)
\]
are increasing in $\lambda(\theta)$ for every $x$.\footnote{In other words, relative to the state space $\{\lambda(\theta):\theta\in\Theta\}$, by Theorem~5.2 in \citet{Lehmann}, $\xi(X_{k+1,n+1})$ is more accurate than $\xi(X_{k,n})$ and $\xi(X_{k,n})$ is more accurate than $\xi(X_{k,n+1})$, .} This is equivalent to both expressions in \eqref{EQ:GG} being increasing in $\theta$ for every $x$, because both $\xi$ and $\lambda$ are strictly increasing.

\medskip

We now show that (i) implies (iii), thus completing the proof of the theorem. Condition (i) says that $X_{k,n}$ is more accurate than $X_{k,n+1}$ for every $k\geq1$ and $n\geq k$. We will prove that, in fact, (iii) is implied by the following weaker condition: $X_{n,n}$ is more accurate than $X_{n,n+1}$ for every $n\geq 1$. We begin with the following lemma, proved in Appendix~\ref{APP:ROCEXP} below.
\begin{lemma}
\label{LEM:ROCEXP}
Assume that $X_{n,n}$ is more accurate than $X_{n,n+1}$ for every $n\geq 1$. Then for every pair of states $\theta'>\theta$ we have $\overline{x}_{\theta'}=\overline{x}_{\theta}$, and there exists $a_{\theta,\theta'}\geq1$ such that
\[
F\big(F^{-1}(p|\theta)\big|\theta'\big)=\max\{0,1-a_{\theta,\theta'}(1-p)\}
\qquad\text{for all $p\in[0,1]$.}
\]
\end{lemma}

By Lemma~\ref{LEM:ROCEXP}, given any three states $\theta''>\theta'>\theta$ we have $a_{\theta,\theta''}=a_{\theta,\theta'}a_{\theta',\theta''}$. Indeed, taking any $x$ such that $0<F(x|\theta'')<1$ we also have $0<F(x|\theta)<1$ and $0<F(x|\theta')<1$ because $\underline{x}_{\theta}\leq\underline{x}_{\theta'}\leq\underline{x}_{\theta''}$ and, by Lemma~~\ref{LEM:ROCEXP}, $\overline{x}_{\theta}=\overline{x}_{\theta'}=\overline{x}_{\theta''}$. Thus, by Lemma~~\ref{LEM:ROCEXP}, we have both $1-F(x|\theta'')=a_{\theta,\theta''}[1-F(x|\theta)]$ and $1-F(x|\theta'')=a_{\theta,\theta'}a_{\theta',\theta''}[1-F(x|\theta)]$, hence $a_{\theta,\theta''}=a_{\theta,\theta'}a_{\theta',\theta''}$.

Now pick any state $\theta_0$ and define
\[
A(\theta)
:=
\begin{cases}
1/a_{\theta,\theta_0} & \text{if $\theta<\theta_0$,} \\
1 & \text{if $\theta=\theta_0$,} \\
a_{\theta_0,\theta} & \text{if $\theta>\theta_0$,}
\end{cases}
\quad\qquad
\lambda(\theta):=\log A(\theta).
\]
The function $\theta\mapsto\lambda(\theta)$ is strictly increasing, because for every $\theta'>\theta$ we have $A(\theta')/A(\theta)=a_{\theta,\theta'}$ and therefore $\lambda(\theta')-\lambda(\theta)=\log a_{\theta,\theta'}$. By \eqref{EQ:SD}, $a_{\theta,\theta'}>1$, so $\lambda(\theta')>\lambda(\theta)$.

Finally, let $\overline{x}$ denote the upper bound that (by Lemma~\ref{LEM:ROCEXP}) is common to all distributions in the family $F$, let $I:=\{x<\overline{x}:F(x|\theta)>0\;\text{for some $\theta\in\Theta$}\}$, and define $\xi:I\to\R$ by
\[
 \xi(x):=\lambda(\theta)-\log\big[1-F(x|\theta)\big],
\]
where $\theta$ is any state such that $0<F(x|\theta)<1$. The definition of $\xi$ is independent of the chosen state, because given any two states $\theta'>\theta$, for every $x$ such that $0<F(x|\theta')<1$ we have
\[
1-F(x|\theta')=e^{\lambda(\theta')-\lambda(\theta)}\big[1-F(x|\theta)\big],
\]
and hence
\[
\lambda(\theta')-\log\big[1-F(x|\theta')\big]
=
\lambda(\theta)-\log\big[1-F(x|\theta)\big].
\]
Moreover, $\xi$ is strictly increasing, because $F(\cdot|\theta)$ is strictly increasing on $(\underline{x}_\theta,\overline{x})$. In each state $\theta$, $\xi(X_i)-\lambda(\theta)=-\log\big[1-F(X_i|\theta)\big]$ has the exponential distribution, hence (iii) follows.

\subsection{Proof of Lemma~\ref{LEM:ROCEXP}\label{APP:ROCEXP}}
Assume $X_{n,n}$ is more accurate than $X_{n,n+1}$ for every $n\geq 1$, and fix any pair of states $\theta'>\theta$. Let
\[
B_{n,n}(p):=1-(1-p)^n,
\qquad
B_{n,n+1}(p):=1-(1-p)^n(1+np)
\]
be the cumulative distribution functions of the lowest order statistic from $n$ independent uniform draws and the second-lowest order statistic from $n+1$ independent uniform draws, respectively. Define $\varphi_n(p):=B^{-1}_{n,n+1}(B_{n,n}(p))$ and $T(p):=F(F^{-1}(p|\theta)|\theta')$ for every $p\in(0,1)$. By \eqref{EQ:MLR}, $T$ is increasing and convex, with $T(p)\leq p$. Since $X_{n,n}$ is more accurate than $X_{n,n+1}$,
\begin{equation}
\label{EQ:phi-T}
\varphi_n(T(p))
\leq
T(\varphi_n(p))
\qquad\text{for all $p\in(0,1)$}.
\end{equation}
Define $s_n(p):=(\varphi_n(p)-p)/(1-p)$. Given that $B_{n,n+1}(p)\leq B_{n,n}(p)$, we have $\varphi_n(p)\geq p$ and hence $1+np\leq 1+n\varphi_n(p)\leq 1+n$, or equivalently
\begin{equation}
\label{EQ:nnp}
1+np \leq 1+np+n(1-p)s_n(p)\leq1+n.
\end{equation}
Moreover, we can write $B_{n,n+1}(\varphi_n(p))=B_{n,n}(p)$ as
\[
(1-s_n(p))^n(1+np+n(1-p)s_n(p))=1.
\]
Taking logs, we obtain
\[
n[-\log(1-s_n(p))]=\log[1+np+n(1-p)s_n(p)]
\]
and hence, by \eqref{EQ:nnp},
\[
\frac{\log(1+np)}{\log n}
\leq
\frac{n[-\log(1-s_n(p))]}{\log n}
\leq
\frac{\log(1+n)}{\log n}.
\]
As $n\to\infty$, both bounds converge to one, hence the middle expression also converges to one. It follows that $s_n(p)$ converges to zero, and since $\lim_{s\to0}[-\log(1-s)]/s=1$, we also have
\[
\frac{n(\varphi_n(p)-p)}{\log n}
=
(1-p)
\frac{ns_n(p)}{\log n}
\to1-p
\qquad\text{as $n\to\infty$.}
\]
Returning to the inequality \eqref{EQ:phi-T}, for every $p\in(0,1)$ such that $T(p)>0$ we therefore have
\[
\frac{T(p+\varphi_n(p)-p)-T(p)}{\varphi_n(p)-p}
\geq
\frac{\varphi_n(T(p))-T(p)}{\varphi_n(p)-p}
\to\frac{1-T(p)}{1-p}
\qquad
\text{as $n\to\infty$}.
\]
Since $s_n(p)$ converges to zero, so does $\varphi_n(p)-p=(1-p)s_n(p)$, hence the above inequality gives $dT(p)/dp\geq(1-T(p))/(1-p)$. But, by convexity of $T$, for every $q>p$ we have $dT(p)/dp\leq(T(q)-T(p))/(q-p)$. Thus, letting $b:=\lim_{q\to1}T(q)=F(\overline{x}_\theta|\theta')\leq1$, we also have $dT(p)/dp\leq(b-T(p))/(1-p)$, and hence $b\geq1$. We conclude that $b=1$, which implies $\overline{x}_{\theta'}=\overline{x}_{\theta}$. Furthermore, letting $c=\sup\{p:T(p)=0\}$, we have $dT(p)/dp=(1-T(p))/(1-p)$ for every $p\in(c,1)$. Solving the latter equation gives $T(p)=1-(1-p)/(1-c)$ for every $p\in(c,1)$, and defining $a_{\theta,\theta'}:=1/(1-c)\geq1$ concludes the proof.

\subsection{Proof of Theorem~\ref{THM:AFA}\label{APP:PROOF-THM3}}

Define $j_n:=n-k_n+1$, $p_n:=j_n/(n+1)$ and $U^\theta_{n}:=F(X_{k_n,n}|\theta)$ for every $n\geq1$ and state $\theta$. Then, at each state $\theta$, the random variable $U^\theta_n$ has the same distribution as the $j_n$-th lowest order statistic in a uniform sample, namely $\text{Beta}(j_n,n-j_n+1)$. Thus,
\begin{equation*}
\E_\theta\big[U^\theta_{n}\big]=p_n
\qquad\text{and}\qquad
\Var_\theta\big(U^\theta_{n}\big)=\frac{p_n(n+1-j_n)}{(n+1)(n+2)}.
\end{equation*}
Applying Chebyshev's inequality to the random variables $U^\theta_n/p_n$ and $(1-U^\theta_n)/(1-p_n)$, we obtain
\begin{align}
\label{EQ:Up1}
j_n\to\infty
&\IMPLIES
\frac{U^\theta_n}{p_n}\tconv1,
\\
\label{EQ:Up2}
k_n\to\infty
&\IMPLIES
\frac{1-U^\theta_n}{1-p_n}\tconv1.
\end{align}

Choose a sequence of sets $\Theta_n\subseteq\Theta$ as follows. If $\Theta$ is finite or compact, let $\Theta_n=\Theta$ for all $n$. Otherwise, let $\Theta_n$ be any sequence of compact intervals such that $\inf\Theta_n\downarrow\inf\Theta$ and $\sup\Theta_n\uparrow\sup\Theta$.

\subsubsection*{Proof of Part (i)}

Assume $k_n\to\infty$ and $j_n\to\infty$. Define $D_n:\R\times\Theta\to\R\cup\{\infty\}$ by
\begin{equation}
\label{EQ:D_n}
D_n(x,\theta):=
\left|\log\frac{F(x|\theta)}{p_n}\right|
+
\left|\log\frac{1-F(x|\theta)}{1-p_n}\right|.
\end{equation}
By \eqref{EQ:Up1} and \eqref{EQ:Up2}, for every $\theta$ we have
\begin{equation}
\label{EQ:Dt}
 D_n(X_{k_n,n},\theta)\tconv0.
\end{equation}
Next, define $\phi_n(x):\R\to\Theta$ as follows:
\begin{equation}
\label{EQ:phi_n}
\phi_n(x):=\arg\min_{\theta\in\Theta_n} D_n(x,\theta).\footnote{The function $\theta\mapsto D_n(x,\theta)$ is lower semicontinuous because $F(x|\theta)$ is continuous in $\theta$. Since $\Theta_n$ is compact, the minimum exists. Moreover, by continuity of $F(x|\theta)$ in both $x$ and $\theta$, a measurable argmin can be selected.}
\end{equation}
Fix $\theta\in\Theta$ and $\ve>0$ now. We must prove that
\begin{equation}
\label{EQ:conv1}
P_\theta(\phi_n(X_{k_n,n})>\theta+\ve)\to0
\AND
P_\theta(\phi_n(X_{k_n,n})<\theta-\ve)\to0.
\end{equation}
Let
\[
\theta_+:=\inf\{\vartheta\in\Theta:\vartheta>\theta+\ve\},
\qquad
\theta_-:=\sup\{\vartheta\in\Theta:\vartheta<\theta-\ve\},
\]
with the convention that $\inf\varnothing=\infty$ and $\sup\varnothing=-\infty$. Then it suffices to prove:\footnote{Our convention makes the hypothesis in \eqref{EQ:conv11} false exactly when $\Theta\subseteq(-\infty,\theta+\ve]$, in which case the first convergence in \eqref{EQ:conv1} is immediate. Similarly, the hypothesis in \eqref{EQ:conv12} is false exactly when $\Theta\subseteq[\theta-\ve,\infty)$, in which case the second convergence in \eqref{EQ:conv1} is immediate.}
\begin{align}
\label{EQ:conv11}
\theta_+\in\Theta
&\IMPLIES
P_\theta(\phi_n(X_{k_n,n})<\theta_+)\to1,
\\
\label{EQ:conv12}
\theta_-\in\Theta
&\IMPLIES
P_\theta(\phi_n(X_{k_n,n})>\theta_-)\to1.
\end{align}
As $\inf\Theta_n\downarrow\inf\Theta$ and $\sup\Theta_n\uparrow\sup\Theta$, for all large $n$ we have $\theta\in\Theta_n$, hence $D_n(x,\theta)\geq D_n(x,\phi_n(x))$. Assuming $\theta_+\in\Theta$ and letting $\Theta^+_n:=\{\vartheta\in\Theta_n:\vartheta\geq\theta_+\}$ we then have, for all $x$, large $n$, and $c>0$,
\[
D_n(x,\theta)
<c<
\inf_{\vartheta\in\Theta^+_n} D_n(x,\vartheta)
\qquad\Longrightarrow\qquad
\phi_n(x)\,<\,\theta_+.
\]
Similarly, assuming $\theta_-\in\Theta$ and letting $\Theta^-_n:=\{\vartheta\in\Theta_n:\vartheta\leq\theta_-\}$, we have
\[
D_n(x,\theta)
<c<
\inf_{\vartheta\in\Theta^-_n} D_n(x,\vartheta)
\qquad\Longrightarrow\qquad
\phi_n(x)\,>\,\theta_-.
\]
Therefore, by \eqref{EQ:Dt}, to prove \eqref{EQ:conv11} and \eqref{EQ:conv12} it suffices to exhibit $c_+>0$ and $c_->0$ such that
\begin{align}
\label{EQ:CCc1}
\theta_+\in\Theta
&\IMPLIES
P_\theta\big(\inf\nolimits_{\vartheta\in\Theta^+_n} D_n(X_{k_n,n},\vartheta)>c_+\big)\to1,
\\
\label{EQ:CCc2}
\theta_-\in\Theta
&\IMPLIES
P_\theta\big(\inf\nolimits_{\vartheta\in\Theta^-_n} D_n(X_{k_n,n},\vartheta)>c_-\big)\to1.
\end{align}
Equivalently, we prove that every subsequence contains a further subsequence along which \eqref{EQ:CCc1} and \eqref{EQ:CCc2} hold. So fix an arbitrary subsequence. Passing to a further subsequence if necessary, we may assume $p_n\to p\in[0,1]$. In what follows, all convergence statements concerning random variables are meant in $P_\theta$-probability.

\bigskip\noindent
\emph{Case 1: $0<p<1$.} By \eqref{EQ:Up1}, $U^\theta_n\to p$, hence $X_{k_n,n}\to F^{-1}(p|\theta)$. Assume $\theta_+\in\Theta$. By continuity, $F(X_{k_n,n}|\theta_+)\to F(F^{-1}(p|\theta)|\theta_+)$, and since $\theta_+>\theta$, by \eqref{EQ:SD} the limit is strictly smaller than $p$. Thus, there exists $0<a<1$ such that $P_\theta(F(X_{k_n,n}|\theta_+)/p_n\leq a)\to1$. By \eqref{EQ:D}, $F(X_{k_n,n}|\theta_+)/p_n\leq a$ implies $F(X_{k_n,n}|\vartheta)/p_n\leq a$ for all $\vartheta\geq\theta_+$. Therefore,
\begin{equation*}
F(X_{k_n,n}|\theta_+)/p_n\leq a
\qquad\Longrightarrow\qquad
\inf_{\vartheta\in\Theta^+_n}
\left\vert\log\frac{F(X_{k_n,n}|\vartheta)}{p_n}\right|\geq -\log a,
\end{equation*}
which implies \eqref{EQ:CCc1} once we choose any $c_+\in(0,-\log a)$. Assume $\theta_-\in\Theta$ now. Then, as before, we have $F(X_{k_n,n}|\theta_-)\to F(F^{-1}(p|\theta)|\theta_-)$, and since $\theta_-<\theta$, by \eqref{EQ:SD} the limit is strictly larger than $p$. Thus, there exists $b>1$ with $P_\theta(F(X_{k_n,n}|\theta_-)/p_n\geq b)\to1$. By \eqref{EQ:D}, $F(X_{k_n,n}|\theta_-)/p_n\geq b$ implies $F(X_{k_n,n}|\vartheta)/p_n\geq b$ for all $\vartheta\leq\theta_-$. Therefore,
\begin{equation*}
F(X_{k_n,n}|\theta_-)/p_n\geq b
\qquad\Longrightarrow\qquad
\inf_{\vartheta\in\Theta^-_n}
\left\vert\log\frac{F(X_{k_n,n}|\vartheta)}{p_n}\right|\geq \log b,
\end{equation*}
and choosing any $c_-\in(0,\log b)$ gives \eqref{EQ:CCc2}.

\bigskip\noindent
\emph{Case 2: $p=0$.} Assume first that $\theta_+\in\Theta$, choose $x_0$ with $0<F(x_0|\theta_+)<1$, and let $a_0:=F(x_0|\theta_+)/F(x_0|\theta)$. Then $0<a_0<1$ by \eqref{EQ:SD}. Since $p_n\to0$ and $U^\theta_n/p_n\to1$, we have $U^\theta_n\to0$, hence $P_\theta(X_{k_n,n}\leq x_0)\to1$.  But $X_{k_n,n}\leq x_0$ implies $F(X_{k_n,n}|\theta_+)/F(X_{k_n,n}|\theta)\leq a_0$ by \eqref{EQ:MLR2}. Thus, choosing any $a\in(a_0,1)$, on the intersection of the events $X_{k_n,n}\leq x_0$ and $U^\theta_n/p_n\leq a/a_0$, whose $P_\theta$-probability tends to one, we have
\[
\frac{F(X_{k_n,n}|\theta_+)}{p_n}
=
\frac{F(X_{k_n,n}|\theta_+)}{F(X_{k_n,n}|\theta)}
\frac{F(X_{k_n,n}|\theta)}{p_n}
\leq
a_0\frac{U^\theta_n}{p_n}
\leq a.
\]
By \eqref{EQ:D}, $F(X_{k_n,n}|\theta_+)/p_n\leq a$ implies $F(X_{k_n,n}|\vartheta)/p_n\leq a$ for all $\vartheta\geq\theta_+$. Therefore,
\begin{equation*}
F(X_{k_n,n}|\theta_+)/p_n\leq a
\qquad\Longrightarrow\qquad
\inf_{\vartheta\in\Theta^+_n}
\left\vert\log\frac{F(X_{k_n,n}|\vartheta)}{p_n}\right|\geq -\log a,
\end{equation*}
which implies \eqref{EQ:CCc1} once we choose any $c_+\in(0,-\log a)$. Assume $\theta_-\in\Theta$ now, choose $x_0$ with $0<F(x_0|\theta)<1$, and let $b_0:=F(x_0|\theta)/F(x_0|\theta_-)$. Then $0<b_0<1$ by \eqref{EQ:SD} and, as before, $P_\theta(X_{k_n,n}\leq x_0)\to1$, while $X_{k_n,n}\leq x_0$ implies $F(X_{k_n,n}|\theta_-)/F(X_{k_n,n}|\theta)\geq 1/b_0$ by \eqref{EQ:MLR2}. Thus, choosing any $b\in(1,1/b_0)$, on the intersection of $X_{k_n,n}\leq x_0$ and $U^\theta_n/p_n\geq bb_0$, whose $P_\theta$-probability tends to one, we have
\[
\frac{F(X_{k_n,n}|\theta_-)}{p_n}
=
\frac{F(X_{k_n,n}|\theta_-)}{F(X_{k_n,n}|\theta)}
\frac{F(X_{k_n,n}|\theta)}{p_n}
\geq
\frac{1}{b_0}\frac{U^\theta_n}{p_n}
\geq b.
\]
By \eqref{EQ:D}, $F(X_{k_n,n}|\theta_-)/p_n\geq b$ implies $F(X_{k_n,n}|\vartheta)/p_n\geq b$ for all $\vartheta\leq\theta_-$. Therefore,
\begin{equation*}
F(X_{k_n,n}|\theta_-)/p_n\geq b
\IMPLIES
\inf_{\vartheta\in\Theta^-_n}
\left\vert\log\frac{F(X_{k_n,n}|\vartheta)}{p_n}\right|\geq \log b,
\end{equation*}
which implies \eqref{EQ:CCc2} once we choose any $c_-\in(0,\log b)$.

\bigskip\noindent
\emph{Case 3: $p=1$.} Assume that $\theta_+\in\Theta$. Choose $x_0$ with $0<F(x_0|\theta)<1$, and let $a_0:=(1-F(x_0|\theta_+))/(1-F(x_0|\theta))$. Then $a_0>1$ by \eqref{EQ:SD}. Moreover, $X_{k_n,n}\geq x_0$ implies $1-F(X_{k_n,n}|\theta_+)\geq a_0(1-F(X_{k_n,n}|\theta))$ by \eqref{EQ:MLR2}. Since $p_n\to1$, by \eqref{EQ:Up2} we have $U^\theta_n\to1$, so $P_\theta(X_{k_n,n}\geq x_0)\to1$. Thus, choosing any $a\in(1,a_0)$, on the intersection of $X_{k_n,n}\geq x_0$ and $a_0(1-U^\theta_n)\geq a(1-p_n)$, whose $P_\theta$-probability tends to one, we have
\[
\frac{1-F(X_{k_n,n}|\theta_+)}{1-p_n}
=
\frac{1-F(X_{k_n,n}|\theta_+)}{1-F(X_{k_n,n}|\theta)}
\frac{1-F(X_{k_n,n}|\theta)}{1-p_n}
\geq
a_0\frac{1-U^\theta_n}{1-p_n}
\geq a.
\]
By \eqref{EQ:D}, $1-F(X_{k_n,n}|\theta_+)\geq a(1-p_n)$ implies $1-F(X_{k_n,n}|\vartheta)\geq a(1-p_n)$ for all $\vartheta\geq\theta_+$. Hence,
\begin{equation*}
1-F(X_{k_n,n}|\theta_+)\geq a(1-p_n)
\IMPLIES
\inf_{\vartheta\in\Theta^+_n}
\left\vert\log\frac{1-F(X_{k_n,n}|\vartheta)}{1-p_n}\right|\geq \log a,
\end{equation*}
which implies \eqref{EQ:CCc1} once we choose any $c_+\in(0,\log a)$. Next, assume $\theta_-\in\Theta$, choose $x_0$ with $0<F(x_0|\theta_-)<1$, and let $b_0:=(1-F(x_0|\theta))/(1-F(x_0|\theta_-))$. Then $b_0>1$ by \eqref{EQ:SD}. Moreover, $X_{k_n,n}\geq x_0$ implies $1-F(X_{k_n,n}|\theta)\geq b_0(1-F(X_{k_n,n}|\theta_-))$ by \eqref{EQ:MLR2}. Since $p_n\to1$, by \eqref{EQ:Up2} we again have $P_\theta(X_{k_n,n}\geq x_0)\to1$. Thus, choosing any $b\in(1/b_0,1)$, on the intersection of $X_{k_n,n}\geq x_0$ and $1-U^\theta_n\leq bb_0(1-p_n)$, whose $P_\theta$-probability tends to one, we have
\[
\frac{1-F(X_{k_n,n}|\theta_-)}{1-p_n}
=
\frac{1-F(X_{k_n,n}|\theta_-)}{1-F(X_{k_n,n}|\theta)}
\frac{1-F(X_{k_n,n}|\theta)}{1-p_n}
\leq
\frac{1}{b_0}\frac{1-U^\theta_n}{1-p_n}
\leq b.
\]
By \eqref{EQ:D}, $1-F(X_{k_n,n}|\theta_-)\leq b(1-p_n)$ implies $1-F(X_{k_n,n}|\vartheta)\leq b(1-p_n)$ for all $\vartheta\leq\theta_-$. Hence,
\begin{equation*}
1-F(X_{k_n,n}|\theta_-)\leq b(1-p_n)
\IMPLIES
\inf_{\vartheta\in\Theta^-_n}
\left\vert\log\frac{1-F(X_{k_n,n}|\vartheta)}{1-p_n}\right|\geq -\log b,
\end{equation*}
which implies \eqref{EQ:CCc2} once we choose any $c_-\in(0,-\log b)$.

\subsubsection*{Proof of Part (ii)}

Assume $n-k_n$ is bounded, and let $J<\infty$ be such that $j_n=n-k_n+1\leq J$ for all $n$. Recall that $p_n=j_n/(n+1)$, so now $p_n\to0$, and since $k_n$ is unbounded, \eqref{EQ:Up2} holds.

\bigskip\noindent
\emph{Sufficiency.} Define $D_n$ and $\phi_n$ as in \eqref{EQ:D_n} and \eqref{EQ:phi_n}. Fix $\theta\in\Theta$ and $\ve>0$, and define $\theta_+$, $\theta_-$, $\Theta^+_n$ and $\Theta^-_n$ as in the proof of part (i). Here, too, we must prove \eqref{EQ:conv11} and \eqref{EQ:conv12}. We will prove that for every $\delta>0$ there exists $M>0$ such that, for every large $n$,
\begin{equation}
\label{EQ:PDM1}
P_\theta(
D_n(X_{k_n,n},\theta)\leq M)>1-\delta,
\end{equation}
while
\begin{equation}
\label{EQ:PDM2}
\theta_+\in\Theta
\IMPLIES
P_\theta\big(
\inf\nolimits_{\vartheta\in\Theta^+_n}
D_n(X_{k_n,n},\vartheta)>M\big)
>1-\delta
\end{equation}
and
\begin{equation}
\label{EQ:PDM3}
\theta_-\in\Theta
\IMPLIES
P_\theta\big(
\inf\nolimits_{\vartheta\in\Theta^-_n}
D_n(X_{k_n,n},\vartheta)>M\big)
>1-\delta.
\end{equation}
Indeed, on the intersection of the events in \eqref{EQ:PDM1} and \eqref{EQ:PDM2}, whose probability is at least $1-2\delta$, the minimizer $\phi_n(X_{k_n,n})$ cannot belong to $\Theta_n^+$. Since $\phi_n(X_{k_n,n})\in\Theta_n$, this implies $\phi_n(X_{k_n,n})<\theta_+$, and since $\delta$ is arbitrary, \eqref{EQ:conv11} follows. Similarly, \eqref{EQ:PDM1} and \eqref{EQ:PDM3} imply \eqref{EQ:conv12}.

To prove \eqref{EQ:PDM1}, choose any $b>0$ and note that, by \eqref{EQ:Up2}, for every large $n$ we have
\begin{equation}
\label{EQ:UUU1}
P_\theta\left(
\left|\log\frac{1-U^\theta_n}{1-p_n}\right|\leq b\right)>1-\frac{\delta}{2}.
\end{equation}
Since $\E_\theta[U^\theta_n/p_n]=1$, Markov's inequality gives $P_\theta(U^\theta_n/p_n>A)\leq 1/A$ for every \mbox{$A>0$}. But, since the $j_n$-th lowest order statistic first-order stochastically dominates the minimum, we also have $P_\theta(U^\theta_n\leq p_n/A)\leq1-\big[1-p_n/A\big]^n\leq np_n/A\leq J/A$. Pick any $A>4J/\delta$. Then, for all $n$, we have both $P_\theta(U^\theta_n/p_n\geq 1/A)>1-\delta/4$ and $P_\theta(U^\theta_n/p_n\leq A)>1-\delta/4$, hence $P_\theta(1/A\leq U^\theta_n/p_n\leq A)>1-\delta/2$, and hence
\begin{equation}
\label{EQ:UUU2}
P_\theta\left(
\left|\log\frac{U^\theta_n}{p_n}\right|\leq \log A\right)>1-\frac{\delta}{2}.
\end{equation}
Combining \eqref{EQ:UUU1} and \eqref{EQ:UUU2} and letting $M:=\log A+b$, we obtain \eqref{EQ:PDM1}.\footnote{Note that in proving \eqref{EQ:PDM1} we have not invoked \eqref{EQ:LI}. Indeed, \eqref{EQ:PDM1} only relies on \eqref{EQ:Up2}. We will use this fact in the proof of Theorem~\ref{THM:CLASS1} below.}

To prove \eqref{EQ:PDM2}, assume $\theta_+\in\Theta$. Applying \eqref{EQ:LI} to the pair $\theta_+>\theta$ and integrating, we have $\inf_x F(x|\theta_+)/F(x|\theta)=0$. Hence, for every $\eta>0$ there exists $x_0$ with $F(x_0|\theta)>0$ such that $F(x_0|\theta_+)/F(x_0|\theta)<\eta$. Since $\E_\theta[U_n^\theta]=p_n\to0$, Markov's inequality gives
\[
P_\theta(X_{k_n,n}>x_0)
=
P_\theta(U_n^\theta>F(x_0|\theta))
\leq
\frac{p_n}{F(x_0|\theta)}
\to0.
\]
By \eqref{EQ:MLR2}, $F(x|\theta_+)/F(x|\theta)$ is increasing in $x$. Hence, on the event $X_{k_n,n}\leq x_0$, we have
\[
\frac{F(X_{k_n,n}|\theta_+)}{F(X_{k_n,n}|\theta)}
\leq
\frac{F(x_0|\theta_+)}{F(x_0|\theta)}
<\eta.
\]
Since $\eta>0$ is arbitrary, we conclude that
\[
\frac{F(X_{k_n,n}|\theta_+)}{F(X_{k_n,n}|\theta)}
\xrightarrow{P_\theta}0.
\]
In particular, for all large $n$,
\[
P_\theta\left(
\frac{F(X_{k_n,n}|\theta_+)}{F(X_{k_n,n}|\theta)}
<\frac{e^{-M}}{A}
\right)>1-\frac{\delta}{2}.
\]
Since we also have $P_\theta(U^\theta_n/p_n\leq A)>1-\delta/4$ for every $n$, and on the intersection of the two events we have
\[
\frac{F(X_{k_n,n}|\theta_+)}{p_n}
=
\frac{F(X_{k_n,n}|\theta_+)}{F(X_{k_n,n}|\theta)}
\frac{U_n^\theta}{p_n}
<e^{-M},
\]
we conclude that, for all large $n$,
\begin{equation}
\label{EQ:Fpe0}
P_\theta\big(
F(X_{k_n,n}|\theta_+)/p_n<e^{-M}
\big)>
1-\delta.
\end{equation}
By \eqref{EQ:D} we have $F(X_{k_n,n}|\theta_+)\geq F(X_{k_n,n}|\vartheta)$ for all $\vartheta\geq\theta_+$. Thus,
\[
F(X_{k_n,n}|\theta_+)/p_n<e^{-M}
\IMPLIES
\sup_{\vartheta\in\Theta^+_n}\frac{F(X_{k_n,n}|\vartheta)}{p_n}<e^{-M},
\]
and hence
\begin{equation}
\label{EQ:FpM0}
F(X_{k_n,n}|\theta_+)/p_n<e^{-M}
\IMPLIES
\inf_{\vartheta\in\Theta_n^+}
\left|
\log\frac{F(X_{k_n,n}|\vartheta)}{p_n}
\right|>M.
\end{equation}
Putting together \eqref{EQ:Fpe0} and \eqref{EQ:FpM0}, we obtain \eqref{EQ:PDM2}.

Finally, to prove \eqref{EQ:PDM3}, assume $\theta_-\in\Theta$. Applying \eqref{EQ:LI} to the pair $\theta>\theta_-$ and integrating, we have $\inf_x F(x|\theta)/F(x|\theta_-)=0$. Hence, for every $\eta>0$ there exists $x_0$ with $F(x_0|\theta)>0$ such that $F(x_0|\theta)/F(x_0|\theta_-)<\eta$. Again by Markov's inequality, $P_\theta(X_{k_n,n}\leq x_0)\to1$, and by \eqref{EQ:MLR2}, $F(x|\theta)/F(x|\theta_-)$ is increasing in $x$, so on the event $X_{k_n,n}\leq x_0$ we have $F(X_{k_n,n}|\theta)/F(X_{k_n,n}|\theta_-)<\eta$. It follows that
\[
\frac{F(X_{k_n,n}|\theta)}{F(X_{k_n,n}|\theta_-)}
\tconv0.
\]
In particular, for all large $n$,
\[
P_\theta\left(
\frac{F(X_{k_n,n}|\theta_-)}{F(X_{k_n,n}|\theta)}
> A e^M
\right)>1-\frac{\delta}{2}.
\]
Since we also have $P_\theta(U^\theta_n/p_n\geq 1/A)>1-\delta/4$ for every $n$, and on the intersection of the two events we have
\[
\frac{F(X_{k_n,n}|\theta_-)}{p_n}
=
\frac{F(X_{k_n,n}|\theta_-)}{F(X_{k_n,n}|\theta)}
\frac{U_n^\theta}{p_n}
>
e^M,
\]
we conclude that, for all large $n$,
\begin{equation}
\label{EQ:Fpe}
P_\theta\big(
F(X_{k_n,n}|\theta_-)/p_n>e^M
\big)>
1-\delta.
\end{equation}
By \eqref{EQ:D}, $F(X_{k_n,n}|\theta_-)/p_n>e^M$ implies $F(X_{k_n,n}|\vartheta)/p_n>e^M$ for all $\vartheta\leq\theta_-$. Thus,
\begin{equation}
\label{EQ:FpM}
F(X_{k_n,n}|\theta_-)/p_n>e^M
\quad\Longrightarrow\quad
\inf_{\vartheta\in\Theta_n^-}
\left|
\log\frac{F(X_{k_n,n}|\vartheta)}{p_n}
\right|>M.
\end{equation}
Putting together \eqref{EQ:Fpe} and \eqref{EQ:FpM}, we obtain \eqref{EQ:PDM3}.

\bigskip\noindent
\emph{Necessity.} Suppose that \eqref{EQ:LI} fails. Then there exist two states $\theta'>\theta$ such that
\[
\inf_x [f(x|\theta')/f(x|\theta)]=:c\in(0,1].
\]
By integration, $F(x|\theta')\geq cF(x|\theta)$, while by \eqref{EQ:D}, $1-F(x|\theta')\geq 1-F(x|\theta)$. Letting $f_{k_n,\,n}(\cdot|\cdot)$ denote the density of $X_{k_n,n}$, we therefore have
\[
\frac{f_{k_n,\,n}(\cdot|\theta')}{f_{k_n,\,n}(\cdot|\theta)}
=
\left[\frac{F(x|\theta')}{F(x|\theta)}\right]^{n-k_n}
\left[\frac{1-F(x|\theta')}{1-F(x|\theta)}\right]^{k_n-1}
\frac{f(x|\theta')}{f(x|\theta)}
\geq c^{n-k_n+1}
\geq c^J
\]
and hence
\begin{equation}
\label{EQ:min}
\int\min\big\{f_{k_n,\,n}(x|\theta),f_{k_n,\,n}(x|\theta')\big\}\,dx \geq c^J.
\end{equation}

Now assume by contradiction that $X_{k_n,n}$ has AFA. Then there exists $\phi_n:\R\to\R$ such that
\[
\phi_n(X_{k_n,n}) \, \tconv\, \theta
\qquad\text{and}\qquad 
\phi_n(X_{k_n,n}) \, \tconvp\, \theta'
\]
and hence, in particular, letting $\bar\theta:=(\theta+\theta')/2$,
\[
 P_\theta\big(\phi_n(X_{k_n,n})>\bar\theta\big)\to0
\qquad\text{and}\qquad 
 P_{\theta'}\big(\phi_n(X_{k_n,n})\leq\bar\theta\big)\to0.
\]
This is impossible, because \eqref{EQ:min} implies that for every $n$ we have
\[
P_\theta\big(\phi_n(X_{k_n,n})>\bar\theta\big)+
P_{\theta'}\big(\phi_n(X_{k_n,n})\leq\bar\theta\big)
\geq
\int\min\big\{f_{k_n,\,n}(x|\theta),f_{k_n,\,n}(x|\theta')\big\}\,dx \geq c^J.
\]

\subsubsection*{Proof of Part (iii)}

This follows from part (ii) by reflection. Reparametrizing the set of states by $\tilde\Theta:=\{-\theta:\theta\in\Theta\}$, the variables $Y_1:=-X_1,\ldots,Y_n:=-X_n$ have distribution $\tilde{F}(y|\tilde\theta)=1-F(-y|-\tilde\theta)$ and density $\tilde{f}(y|\tilde\theta)=f(-y|-\tilde\theta)$. Given any two states $\tilde\theta'>\tilde\theta$ we have $-\tilde\theta'<-\tilde\theta$, hence $\tilde{f}(y|\tilde\theta')/\tilde{f}(y|\tilde\theta)=f(-y|-\tilde\theta')/f(-y|-\tilde\theta)$ is increasing in $y$ by \eqref{EQ:MLR}. Thus, the reflected distribution also satisfies \eqref{EQ:MLR}, while \eqref{EQ:ID} is obviously preserved, too. Moreover, condition \eqref{EQ:LI} for the reflected distribution is
\[
\tilde\theta'>\tilde\theta
\qquad
\Longrightarrow
\qquad
\inf_y \;
\frac{\tilde f(y|\tilde\theta')}{\tilde f(y|\tilde\theta)}
=0,
\]
which is equivalent to \eqref{EQ:UI} for the original distribution. Letting $j_n:=n-k_n+1$, the sequence $k_n$ is bounded if and only if $n-j_n$ is bounded, and the reflected $j_n$-th highest order statistic is $Y_{j_n,n}:=-X_{k_n,n}$. By part (ii), $Y_{j_n,n}$ has AFA if and only if \eqref{EQ:LI} holds for the reflected distribution. Since observing $Y_{j_n,n}$ is equivalent to observing $X_{k_n,n}$, this is equivalent to saying that $X_{k_n,n}$ has AFA if and only if \eqref{EQ:UI} holds for the original distribution.

\subsection{Proof of Theorem~\ref{THM:CLASS1}}

\subsubsection*{Proof of Part (i)}

Fix $r\geq1$. Let $p_n:=r/(n+1)$. Define $D_n$ and $\phi_n$ as in \eqref{EQ:D_n} and \eqref{EQ:phi_n}. Fix a lower class $L$, a state $\theta\in L$, and $\ve>0$. Let $L^\ve:=\{\vartheta\in\Theta:\inf_{\vartheta'\in L}|\vartheta'-\vartheta|<\ve\}$. Define
\[
\theta_+:=\inf\{\vartheta\in\Theta:\vartheta>\sup L^\ve\},
\qquad
\theta_-:=\sup\{\vartheta\in\Theta:\vartheta<\inf L^\ve\},
\]
with the conventions $\inf\varnothing=\infty$ and $\sup\varnothing=-\infty$. Finally, let
\[
\Theta^+_n:=\{\vartheta\in\Theta_n:\vartheta\geq\theta_+\},
\qquad
\Theta^-_n:=\{\vartheta\in\Theta_n:\vartheta\leq\theta_-\}.
\]
Note that if $\inf_{\vartheta\in L}|\phi_n(X_{n-r+1,n})-\vartheta|>\ve$ then either $\theta_+\in\Theta$ and $\phi_n(X_{n-r+1,n})\in\Theta^+_n$, or $\theta_-\in\Theta$ and $\phi_n(X_{n-r+1,n})\in\Theta^-_n$. Thus, to prove the desired convergence, it is enough to show
\begin{align}
\label{EQ:conv11bis}
\theta_+\in\Theta
&\IMPLIES
P_\theta(\phi_n(X_{n-r+1,n})<\theta_+)\to1,
\\
\label{EQ:conv12bis}
\theta_-\in\Theta
&\IMPLIES
P_\theta(\phi_n(X_{n-r+1,n})>\theta_-)\to1.
\end{align}
The proof now reuses the sufficiency argument in the proof of part (ii) of Theorem~\ref{THM:AFA}. The only difference is that instead of \eqref{EQ:LI}, which here may fail to hold globally for all pairs of states, we use the fact that the equality in \eqref{EQ:LI} does hold locally for the pair $\theta_+>\theta$ whenever $\theta_+\in\Theta$, because $\theta$ and $\theta_+$ are in different lower classes---and likewise for $\theta$ and $\theta_-$. Since $n-r+1\to\infty$, the argument proving \eqref{EQ:PDM1} applies (setting $k_n:=n-r+1$) without change: for every $\delta>0$ there exists $M>0$ such that, for all large $n$,
\begin{equation}
\label{EQ:PDM1bis}
P_\theta(
D_n(X_{n-r+1,n},\theta)\leq M)>1-\delta.
\end{equation}
If $\theta_+\in\Theta$, then $\theta_+>\theta$ and $\theta_+\notin L$, so $\inf_x F(x|\theta_+)/F(x|\theta)=0$. Thus, the same argument used to  prove \eqref{EQ:PDM2} in Theorem~\ref{THM:AFA} here gives
\begin{equation}
\label{EQ:PDM2bis}
\theta_+\in\Theta
\IMPLIES
P_\theta\big(
\inf\nolimits_{\vartheta\in\Theta^+_n}
D_n(X_{n-r+1,n},\vartheta)>M\big)
>1-\delta.
\end{equation}
Similarly, $\theta_-\in\Theta$ implies $\theta_-<\theta$ and $\theta_-\notin L$, so $\inf_x F(x|\theta)/F(x|\theta_-)=0$. The same argument used to  prove \eqref{EQ:PDM3} then gives
\begin{equation}
\label{EQ:PDM3bis}
\theta_-\in\Theta
\IMPLIES
P_\theta\big(
\inf\nolimits_{\vartheta\in\Theta^-_n}
D_n(X_{n-r+1,n},\vartheta)>M\big)
>1-\delta.
\end{equation}
Since $\theta\in\Theta_n$ for all large $n$, and moreover $\phi_n(x)$ minimizes $D_n(x,\cdot)$ over $\Theta_n$, we conclude that $D_n(X_{n-r+1,n},\phi_n(X_{n-r+1,n}))\leq D_n(X_{n-r+1,n},\theta)$ for all large $n$. Thus, \eqref{EQ:PDM1bis} and \eqref{EQ:PDM2bis} imply \eqref{EQ:conv11bis}. Similarly, \eqref{EQ:PDM1bis} and \eqref{EQ:PDM3bis} imply \eqref{EQ:conv12bis}, so we are done.

\subsubsection*{Proof of Part (ii)}

Fix a lower class $L$. Since in state $\theta$ the random variable $U^\theta_{n}:=F(X_{n-r+1,n}|\theta)$ has the same distribution as the $r$-th lowest order statistic in a uniform sample, by Theorem~2.2.1 in \citet{Leadbetter} we have
\begin{equation}
\label{EQ:EEE}
P_\theta\big(nU^\theta_n\leq t\big)
\to
1-e^{-t} \sum^{r-1}_{j=0} \frac{t^j}{j!},
\end{equation}
that is, the law of $nU^\theta_n$ converges weakly to the law of a Gamma with shape parameter $r$. This is the law of the sum of $r$ independent exponentials, so we can write \eqref{EQ:EEE} as
\begin{equation}
\label{EQ:EE}
nU^\theta_n \tconvd E_1+\cdots+E_{r},
\end{equation}
where $E_1,\ldots,E_r$ are independent exponentials. Now, given that $\theta$ and $\theta_{\!L}$ belong to the same lower class $L$, we have
\[
\lim_{x\to\underline{x}_{L}}
\frac{F(x|\theta_{\!L})}{F(x|\theta)}=e^{\lambda(\theta|L)}.
\]
By \eqref{EQ:EEE} we have $F(X_{n-r+1,n}|\theta)\to0$ and hence
\[
X_{n-r+1,n}\tconv\underline{x}_{L}.
\]
Thus, by \eqref{EQ:EE} and Slutsky's theorem,
\begin{equation*}
nF(X_{n-r+1,n}|\theta_{\!L}) \tconvd e^{\lambda(\theta|L)}\big(E_1+\cdots+E_{r}\big).
\end{equation*}
Applying the continuous mapping theorem (with respect to the map $x\mapsto\log x$) gives the result.

\subsection{Proofs of Theorems~\ref{THM:ACCORDION2}, \ref{THM:EXP2} and \ref{THM:CLASS2}}
\label{APP:PROOF-STAR}
The proofs obtain from Theorems~\ref{THM:ACCORDION1}, \ref{THM:EXP1} and \ref{THM:CLASS1} by reflection. As in the proof of part (iii) of Theorem~\ref{THM:AFA}, we reparametrize the set of states by $\tilde\Theta:=\{-\theta:\theta\in\Theta\}$, so that $Y_1:=-X_1,\ldots,Y_n:=-X_n$ have distribution $\tilde{F}(y|\tilde\theta)=1-F(-y|-\tilde\theta)$ and density $\tilde{f}(y|\tilde\theta)=f(-y|-\tilde\theta)$. As we argued in the proof of part (iii) of Theorem~\ref{THM:AFA}, the reflected distribution satisfies \eqref{EQ:MLR}, and also \eqref{EQ:ID} if the original distribution satisfies \eqref{EQ:ID}.

\subsubsection{Proof of Theorem~\ref{THM:ACCORDION2}}
The hazard function of the reflected experiment is
\[
\tilde H(y|\tilde\theta)
:=-\log\big[1-\tilde F(y|\tilde\theta)\big]
=-
\log F(-y|-\tilde\theta)
=R(-y|-\tilde\theta).
\]
The mapping $(y,\tilde\theta)\mapsto(-y,-\tilde\theta)$ reverses both coordinates. Therefore $\tilde H$ is log-supermodular in the reflected experiment if and only if $R$ is log-supermodular in the original experiment. By Theorem~\ref{THM:ACCORDION1} applied to the reflected experiment, more minimal selection increases accuracy if and only if $\tilde H$ is log-supermodular. Since $\tilde X_{j+1,n+1}=-X_{n-j+1,n+1}$ and $\tilde X_{j,n}=-X_{n-j+1,n}$, this comparison is exactly more maximal selection increasing accuracy for the original experiment. Hence more maximal selection increases accuracy in the original experiment if and only if $R$ is log-supermodular.

\subsubsection{Proof of Theorem~\ref{THM:EXP2}}
Applying Theorem~\ref{THM:EXP1} to the reflected experiment, more minimal selection decreases accuracy in the original experiment if and only if more maximal selection decreases accuracy in the reflected experiment. Theorem~\ref{THM:EXP1} therefore implies that this is equivalent to the existence of strictly increasing functions $\tilde\lambda$ and $\tilde\xi$ such that
\[
\tilde\xi(\tilde X_i) \tdist \tilde\lambda(\tilde\theta)+E_i,
\]
where $E_i$ is standard exponential. Substituting $\tilde X_i=-X_i$ and $\tilde\theta=-\theta$ gives
\[
\tilde\xi(-X_i) \tdist \tilde\lambda(-\theta)+E_i.
\]
Define
\[
\xi(x):=-\tilde\xi(-x),
\qquad
\lambda(\theta):=-\tilde\lambda(-\theta).
\]
Both mappings are strictly increasing. Multiplying the preceding displayed equation by $-1$ gives
\[
\xi(X_i) \tdist \lambda(\theta)-E_i.
\]
Finally, Theorem~\ref{THM:EXP1} says that, in the reflected experiment, more maximal selection decreasing accuracy is equivalent to the whole more minimal selection increasing and more maximal selection decreasing accuracy. Reflecting this statement back to the original experiment gives that more minimal selection decreases and more maximal selection increases accuracy.

\subsubsection{Proof of Theorem~\ref{THM:CLASS2}}

Observe that if $U$ is an upper class in the original experiment, then $\tilde L:=\{-\theta:\theta\in U\}$ is a lower class in the reflected experiment. Its lower bound is $\underline{\tilde x}_{\tilde L}=-\overline x_U$. We first prove part (i). By part (i) of Theorem~\ref{THM:CLASS1}, applied to the reflected experiment, there exists a sequence of functions $\tilde \phi_n:\R\to\tilde\Theta$ such that, for every lower class $\tilde L$ and every $\tilde\theta\in\tilde L$,
\[
\inf_{\tilde\vartheta\in\tilde L}
\left|
\tilde\phi_n(\tilde X_{n-r+1,n})-\tilde\vartheta
\right|
\, \xrightarrow{\; P_{\tilde\theta} \;} \,0.
\]
Define $\phi_n(x):=-\tilde\phi_n(-x)$. Then, for every upper class $U$ and every $\theta\in U$,
\[
\inf_{\vartheta\in U}
\left|
\phi_n(X_{r,n})-\vartheta
\right|
=
\inf_{\tilde\vartheta\in\tilde L}
\left|
\tilde\phi_n(\tilde X_{n-r+1,n})-\tilde\vartheta
\right|
\tconv 0,
\]
which proves part (i). We now prove part (ii). Fix an upper class $U$ and let $\tilde\theta_{\tilde L}:=-\theta_U$ be the corresponding reference state in the reflected lower class $\tilde L$. The statistic in Theorem~\ref{THM:CLASS1} for the reflected experiment is
\[
\tilde W^{\tilde L}_{r,n}
:=
\log\left(
n\tilde F(\tilde X_{n-r+1,n}|\tilde\theta_{\tilde L})
\right).
\]
Using $\tilde X_{n-r+1,n}=-X_{r,n}$, we obtain $\tilde W^{\tilde L}_{r,n}=\log(n[1-F(X_{r,n}|\theta_U)])=-Y^U_{r,n}$. By part (ii) of Theorem~\ref{THM:CLASS1}, applied to the reflected experiment,
\[
\tilde W^{\tilde L}_{r,n}
\tconvd
\tilde\lambda(\tilde\theta|\tilde L)+Z_r,
\]
where
\[
\tilde\lambda(\tilde\theta|\tilde L)
=
\lim_{y\to \underline{\tilde x}_{\tilde L}}
\left[
\tilde R(y|\tilde\theta)
-
\tilde R(y|\tilde\theta_{\tilde L})
\right].
\]
Since $\tilde R(y|\tilde\theta)=-\log\tilde F(y|\tilde\theta)=-\log[1-F(-y|-\tilde\theta)]=H(-y|-\tilde\theta)$, and since $y\to -\overline x_U$ is equivalent to $x=-y\to\overline x_U$, we have $\tilde\lambda(-\theta|\tilde L)=-\upsilon(\theta|U)$, and hence
\[
Y^U_{r,n}
\tconvd
\upsilon(\theta|U)-Z_r.
\]

\subsection{Proof of Theorems~\ref{THM:BLACKWELL1} and \ref{THM:BLACKWELL2}\label{APP:PROOF-BLACKWELL}}

Since $E_1+\cdots+E_r$ and $E_{r+1}$ are independent and have distributions $\text{Gamma}(r,1)$ and $\text{Gamma}(1,1)$, respectively,
\[
B_r:=\frac{E_1+\cdots+E_r}{E_1+\cdots+E_r+E_{r+1}}=\frac{\exp(Z_r)}{\exp(Z_{r+1})}
\]
has distribution $\text{Beta}(r,1)$ and is independent of $E_1+\cdots+E_r+E_{r+1}$ and hence of $Z_{r+1}$. Taking logs, we obtain $Z_r=Z_{r+1}+\log B_r$, which gives the result: starting from the observation $\lambda(\theta|L)+Z_{r+1}$, adding the independent noise $\log B_r$ gives an observation distributed as $\lambda(\theta|L)+Z_r$. Similarly, starting from $\upsilon(\theta|U)-Z_{r+1}$, adding the independent noise $-\log B_r$ gives an observation distributed as $\upsilon(\theta|U)-Z_r$. To prove the second statement in the theorem, note that, by the law of large numbers, $(E_1+\cdots+E_r)/r$ converges in probability to one. Therefore, $Z_r-\log r=\log((E_1+\cdots+E_r)/r)$ converges in probability to zero.

\subsection{Proof of Theorem~\ref{THM:MULTI}}
\label{APP:PROOF-MULTI}

We prove the more-maximal-selection statement first. Fix $(x_1,\ldots,x_m)\in\R^m$ and a pair of states $\theta'>\theta$. Define $\varphi(\theta,x_1,\ldots,x_m)=(y_1,\ldots,y_m)$ as in \eqref{EQ:multi1} and \eqref{EQ:multi2}, that is, $F_{k,n+1}(y_1|\theta)=F_{k,n}(x_1|\theta)$ and, for $\ell=1,\ldots,m-1$,
\begin{equation}
\label{EQ:sss1}
\left[\frac{F(y_{\ell+1}|\theta)}{F(y_{\ell}|\theta)}\right]^{n-k-\ell+2}
=\;
\left[\frac{F(x_{\ell+1}|\theta)}{F(x_{\ell}|\theta)}\right]^{n-k-\ell+1}.
\end{equation}
Define $\varphi(\theta',x_1,\ldots,x_m)=(y'_1,\ldots,y'_m)$ analogously for the higher state, so that $F_{k,n+1}(y'_1|\theta')=F_{k,n}(x_1|\theta')$ and, for $\ell=1,\ldots,m-1$,
\begin{equation}
\label{EQ:sss1bis}
\left[\frac{F(y'_{\ell+1}|\theta')}{F(y'_{\ell}|\theta')}\right]^{n-k-\ell+2}
=\;
\left[\frac{F(x_{\ell+1}|\theta')}{F(x_{\ell}|\theta')}\right]^{n-k-\ell+1}.
\end{equation}
We have to prove $y'_\ell\geq y_\ell$ for $\ell=1,\ldots,m$. Since log-supermodularity is preserved by integration, $R$ is log-supermodular because so is $r$, hence Theorem~\ref{THM:ACCORDION2} implies $y'_1\geq y_1$. Moreover, since $X_{k,n+1}$ first-order stochastically dominates $X_{k,n}$, we have $y_1\geq x_1$.

Now suppose, inductively, that for some $\ell\geq 1$ we have $y'_{\ell}\geq y_{\ell}\geq x_{\ell}$. By \eqref{EQ:sss1} and $y_{\ell}\geq x_{\ell}$, we have $y_{\ell+1}\geq x_{\ell+1}$. To prove $y'_{\ell+1}\geq y_{\ell+1}$ it suffices to show that replacing $y'_{\ell+1}$ with $y_{\ell+1}$ in  \eqref{EQ:sss1bis} makes the left-hand side smaller, that is,
\begin{equation*}
\left[\frac{F(y_{\ell+1}|\theta')}{F(y'_{\ell}|\theta')}\right]^{n-k-\ell+2}
\leq\;
\left[\frac{F(x_{\ell+1}|\theta')}{F(x_{\ell}|\theta')}\right]^{n-k-\ell+1}.
\end{equation*}
But $y'_\ell\geq y_\ell$ and hence, in fact, it is enough to prove
\begin{equation}
\label{EQ:sss3}
\left[
\frac{F(y_{\ell+1}|\theta')}{F(y_\ell|\theta')}
\right]^{n-k-\ell+2}
\leq
\left[
\frac{F(x_{\ell+1}|\theta')}{F(x_\ell|\theta')}
\right]^{n-k-\ell+1}.
\end{equation}
Equivalently, taking logs, we shall prove
\begin{equation}
\label{EQ:sss6}
(n-k-\ell+2)\int_{y_{\ell+1}}^{y_\ell} r(y|\theta')\,dy
\ge
(n-k-\ell+1)\int_{x_{\ell+1}}^{x_\ell} r(y|\theta')\,dy .
\end{equation}

For every $x\in[x_{\ell+1},x_\ell]$, define $v(x)$ by
\begin{equation}
\label{EQ:sss4}
\left[\frac{F(v(x)|\theta)}{F(y_\ell|\theta)}\right]^{n-k-\ell+2}
=\;
\left[\frac{F(x|\theta)}{F(x_\ell|\theta)}\right]^{n-k-\ell+1}.
\end{equation}
Thus $v(x_\ell)=y_\ell$, while $v(x_{\ell+1})=y_{\ell+1}$. Moreover, since $y_\ell\geq x_\ell$, the maximum of $n-k-\ell+2$ draws from the distribution truncated from above at $y_\ell$ first-order stochastically dominates the maximum of $n-k-\ell+1$ draws from the distribution truncated from above at $x_\ell$. Hence $v(x)\geq x$ for all $x\in[x_{\ell+1},x_\ell]$. Taking logs in \eqref{EQ:sss4}, we have
\[
(n-k-\ell+2)\int_{v(x)}^{y_\ell} r(y|\theta)\,dy
=
(n-k-\ell+1)\int_x^{x_\ell} r(y|\theta)\,dy
\]
and hence, differentiating with respect to $x$,
\begin{equation}
\label{EQ:sss5}
(n-k-\ell+2)r(v(x)|\theta)v'(x)
=
(n-k-\ell+1)r(x|\theta).
\end{equation}
Thus, using the change of variable $y=v(x)$, the left-hand side of \eqref{EQ:sss6} equals
\[
(n-k-\ell+2)\int_{x_{\ell+1}}^{x_\ell} r(v(x)|\theta')v'(x)\,dx
=
(n-k-\ell+1)\int_{x_{\ell+1}}^{x_\ell}
r(x|\theta)
\frac{r(v(x)|\theta')}{r(v(x)|\theta)}
\,dx .
\]
Since $r$ is log-supermodular and $v(x)\geq x$, we have
\[
\frac{r(v(x)|\theta')}{r(v(x)|\theta)}
\ge
\frac{r(x|\theta')}{r(x|\theta)},
\]
hence \eqref{EQ:sss6} follows.

The more-minimal-selection statement follows by reflection. The reverse hazard rate of the reflected experiment is
\[
\tilde r(y|\tilde\theta)
=
\frac{\tilde f(y|\tilde\theta)}{\tilde F(y|\tilde\theta)}
=
\frac{f(-y|-\tilde\theta)}{1-F(-y|-\tilde\theta)}
=
h(-y|-\tilde\theta).
\]
Since the map $(y,\tilde\theta)\mapsto(-y,-\tilde\theta)$ reverses both coordinates, log-supermodularity of $h$ in the original experiment is equivalent to log-supermodularity of $\tilde r$ in the reflected experiment. More minimal selection in the original experiment corresponds, after reflection, to more maximal selection in the reflected experiment. Therefore, applying the more-maximal result just proved to the reflected experiment, if $h$ is log-supermodular then $\mathbf X^m_{k+1,n+1}$ is more accurate than $\mathbf X^m_{k,n}$.

\phantomsection
\addcontentsline{toc}{section}{References}

\bibliographystyle{bibstyle}
\bibliography{Info_Order_Stats}

\end{document}